\documentclass[11pt,reqno]{amsart}
\usepackage[left=1in,right=1in]{geometry}
\usepackage{bm, makecell}
\usepackage{float}
\restylefloat{table}
\usepackage{amssymb,latexsym}
\usepackage{color}
\usepackage{graphicx}
\usepackage{epstopdf}
\usepackage{subfig}
\usepackage[mathcal]{euscript}
\usepackage{mathrsfs,amsmath}
\usepackage{enumerate}
\usepackage{hyperref}
\usepackage{amsfonts}
\usepackage{amssymb}
\usepackage{tikz}
\usetikzlibrary{shapes.geometric, arrows}
\usepackage{cite}
\newtheorem{thm}{Theorem}[section]

\theoremstyle{definition}

\title{Assessing the Effects of Treatment in HIV-TB Co-infection Model}

\author[S. Kumar]{Sachin Kumar}
\author[S. Jain]{Shikha Jain$^\ast$} \thanks{$^\ast$Corresponding Author}

\address{Sachin Kumar,
Assistant Professor,
Department of Mathematics,
University of Delhi,
New Delhi-110007,
India.}
\email{sachinambariya@gmail.com}

\address{Shikha Jain, 
Department of Mathematics,
University of Delhi,
Delhi-110007,
India.}
\email{shikhajain01051990@gmail.com}



\numberwithin{equation}{section}
\allowdisplaybreaks

\begin{document}

\begin{abstract} We propose a population model for HIV-TB co-infection dynamics by considering treatments for HIV infection, active tuberculosis and co-infection. The HIV only and TB only models are analyzed separately, as well as full model. The basic reproduction numbers for TB ($\mathcal{R}_0^T$) and HIV ($\mathcal{R}_0^H$) and overall reproduction number for the system $\mathcal{R}_0= \max\{\mathcal{R}_0^T, \mathcal{R}_0^H\}$ are computed. The equilibria and their stability are studied. The main model undergoes supercritical transcritical bifurcation at $\mathcal{R}_0^T=1$ and $\mathcal{R}_0^H=1$ whereas the parameters $\beta^*=\beta e$ and $\lambda^*=\lambda \sigma$ act as bifurcation parameters, respectively. Numerical simulation claims the existence of interior equilibrium when both the reproduction numbers are greater than unity. We explore the effect of early and late HIV treatment on disease-induced deaths during the TB treatment course. Mathematical analysis of our model shows that successful disease eradication requires treatment of single disease, that is, treatment for HIV only and TB only infected individuals with addition to co-infection treatment and in absence of which disease eradication is extremely difficult even for $\mathcal{R}_0<1$. When both the diseases are epidemic, the treatment for TB only infected individuals is very effective in reducing the total infected population and disease-induced deaths in comparison to the treatment for HIV infected individuals while these are minimum when both the single disease treatments are given with co-infection treatment.

\end{abstract}

\maketitle

\section{\textbf{Introduction}}
According to WHO, Tuberculosis is one of the top 10 causes of death worldwide \cite{whotb, TB}. Tuberculosis (TB) is a bacterial disease which is primarily caused by the bacteria \emph{Mycobacterium Tuberculosis} and is usually acquired by inhaling TB bacteria from surrounding air. All the infected people are not equally infectious and generally, it is only people with TB of the throat or lungs who are infectious. The bacteria get released in air by a carrier with active TB through coughing, sneezing or talking. Most infections do not have symptoms, in that case it is known as latent tuberculosis and people with latent TB do not spread the disease. Inhaling only a few of these germs are sufficient to get infected. \textit{About one quarter of the world's population is infected with TB} \cite{TBt}, while most are in latent phase.\par
HIV, the Human Immunodeficiency Virus infects cells of immune system, destroying their function. HIV infects vital cells in the human immune system such as helper T cells (specifically CD4+ T cells), macrophages, and dendritic cells. As it hijacks the T cells that help keep the immune system working, HIV is particularly devastating to immune health. In the process of replication, the virus destroys increasing numbers of T cells. The T cells of an important part of the immune system are annihilated, leaving the body open to opportunistic infections. The immune system is thus deteriorated and no longer fulfils its role of fighting infections and diseases. Acquired immunodeficiency syndrome (AIDS) is a term for the most advanced stages of HIV infection. HIV can be transmitted through unprotected sexual intercourse, and oral sex with an infected person, transfusion of contaminated blood and sharing of contaminated needles, syringes, surgical equipment or other sharp instruments. It may also be transmitted between a mother and her infant during pregnancy, childbirth and breastfeeding. In 2015, an estimated 44\% of new infections occurred among key populations and their partners \cite{HIVf}.

TB is a leading killer of HIV-positive people: in $2015$, 35\% of HIV deaths were due to TB \cite{TBt}. Lowered immunity due to HIV infection increases the susceptibility to TB infection. People infected with HIV are 20 to 30 times more likely to develop active TB disease than the uninfected ones. TB is a treatable and curable disease but it is important to complete the entire course of medications even after one feels well. Between 2000 to 2015, an estimated 49 million lives were saved through TB diagnosis and treatment \cite{TBt}. Though HIV infection has no permanent treatment, ART (Antiretroviral Therapy) can slow down the progression of HIV in the body to near a halt. ART reduces the risk of TB morbidity and mortality among people living with HIV. When ART is combined with TB preventive therapy, it can have a significant impact on TB prevention. Since TB can be cured effectively with treatment and complete treatment course being short, the usual recommendation is to start it immediately. The DOTS strategy makes no distinction between settings with different levels of HIV infection, yet outcomes will inevitably differ according to the epidemiology of HIV infection \cite{DOTS}. Initiating ART soon after the beginning of TB treatment increases the risk of IRIS (Immune Reconstruction Inflammatory Syndrome) which worsens TB infection and causes severe medical complications, while its delay until completion of TB treatment course increases the risk of death due to HIV. Therefore, it is difficult to identify the correct initiation of ART with TB treatment.

The negative impact of synergic interactions between TB and HIV have caused worldwide concern. Mathematical modelling of HIV, TB and HIV/TB co-infections have been reported by several researchers. Guzzetta et. al \cite{Guzzetta} proposed an age-structured, socio-demographic individual based model (IBM) with a realistic, time-evolving structure of preferential contacts in a population. Trauer et. al \cite{Trauer}  presented a mathematical model to simulate tuberculosis (TB) transmission in highly endemic regions of the Asia-Pacific, where epidemiology does not appear to be primarily driven by HIV-coinfection. Long et. al \cite{Elisa} proposed a co-epidemical model for HIV-TB infection and presented an analysis in the population of India. Roeger et. al \cite{FengChavez} proposed an 8 compartmental model of HIV-TB co-infection in which qualitative analysis of the model has been done. They discussed the stability and disease prevalence in the model. Silva et. al \cite{Silva} proposed a population model for HIV-TB/AIDS co-infection transmission dynamics, which considers antiretroviral therapy for HIV infection and treatments for latent and active tuberculosis. Bhunu et. al \cite{Bhunnu} developed a model that in-corporates all aspects of TB transmission dynamics as well as aspects of HIV transmission dynamics to come with a distinct detailed co-infection model for HIV and TB. Naresh et. al \cite{NT} developed a HIV/TB co-epidemic model assuming that AIDS cases are non- infectious and did not include all stages of HIV and TB infection. Gakkhar and Chavda \cite{Gakkhar} formulated a simple epidemic HIV-TB co-infection model. Kaur et. al \cite{Kaur} developed a deterministic non-linear HIV-TB co-infection model which discusses the role of screening and treatment in the transmission dynamics of HIV/AIDS and tuberculosis co-infection. Mallela et. al \cite{mallela} developed an eight compartmental model and studied the effect of HIV treatment in different phases of TB treatment.
 
 Work done by Mallela et. al is the motivation for the present paper. In the present paper necessity of single disease infection treatments that is treatment for TB only and HIV only patients is also studied with the HIV treatment during different phases of TB treatment.
 
  The paper is organised as follows: in Section \ref{Section2}, a twelve compartmental model for HIV-TB co-infection and treatment has been developed and positivity and boundedness of the solutions is proved. In Sections \ref{Section3} and \ref{Section4}, TB and HIV sub-models are analyzed, respectively, and the respective reproduction numbers are calculated. The existence and stability conditions of equilibria are also deduced. In Section \ref{Section5}, the main model is discussed with its reproduction number and stability of equilibria. In Section \ref{Section6}, numerical computations of the model are performed to explore the HIV-TB co-infection dynamics. The effect of reproduction number on the infected population is studied. The effect of early or late initiation of HIV treatment during TB treatment course and the necessity of single disease infection treatment are discussed. We summarize our results with conclusion in Section \ref{Section7}.

\section{\textbf{Model formulation and basic properties}}\label{Section2}
The model subdivides the human population into twelve mutually-exclusive compartments, namely susceptible individuals ($S$), TB-latently infected individuals ($T_L$), TB-infected individuals who are infectious and have active TB ($T_I$), TB-infected individuals who are under treatment for TB ($T_T$), HIV infected individuals ($H$), HIV infected individuals co-infected with latent TB ($H_L$), HIV infected individuals under treatment for HIV infection ($H^T$), co-infected Individuals with active TB ($C$), co-infected individuals with active TB under early phase of treatment for TB ($C_1$), co-infected individuals with active TB under late phase of treatment for TB ($C_2$), co-infected individuals with active TB under early phase of TB treatment as well as going through ART ($C_1^T$) and co-infected individuals with active TB under late phase of TB treatment along with ART ($C_2^T$). The total population at time t, denoted by $N(t)$, is given by 
\begin{align*}
N(t)=& S(t)+T_L(t)+T_I(t)+H(t)+H_L(t)+H^T(t)+C(t)\\ &+C_1(t)+C_2(t)+C_1^T(t)+C_2^T(t)+T_T(t).
\end{align*}
\tikzstyle{startstop}=[rectangle, rounded corners]
\tikzstyle{block}=[draw,fill=pink!50,text width=5em,text centered,node distance=9em,minimum height=15mm]
\tikzstyle{arrow}=[thick,->,>=stealth]

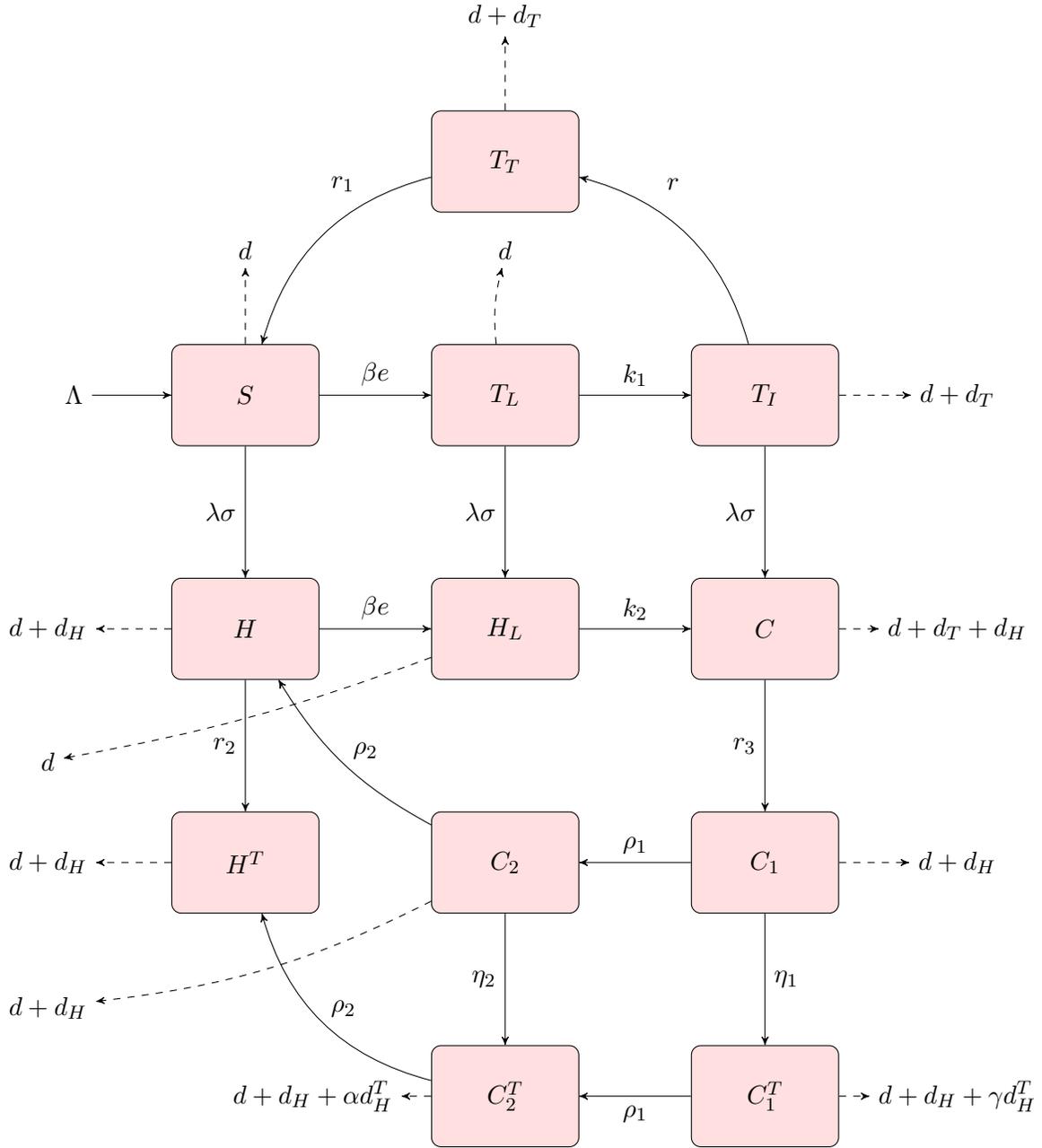
\begin{figure}
\begin{center}

\begin{tikzpicture}[->,>=stealth']

\node[block,rounded corners](TT) {$T_T$};
\node[block,rounded corners,below of=TT,xshift=-10em](S) {$S$};
\node[block,rounded corners,below of=TT,xshift=0em](TL) {$T_L$};
\node[block,rounded corners,below of=TT,xshift=10em](TI) {$T_I$};
\node[block,rounded corners,below of=S](H) {$H$};
\node[block,rounded corners,below of=TL](HL) {$H_L$};
\node[block,rounded corners,below of=TI](C) {$C$};
\node[block,rounded corners,below of=C,xshift=-10em](C2) {$C_2$};
\node[block,rounded corners,below of=C](C1) {$C_1$};
\node[block,rounded corners,below of=C2](C2T) {$C_2^T$};
\node[block,rounded corners,below of=C1](C1T) {$C_1^T$};
\node[block,rounded corners,below of=H](HT) {$H^T$};

\node [text centered,left of=S,xshift=-4em] (Lambda) {$\Lambda$};
\node [text centered,above of=S,yshift=3em] (d) {$d$};
\node [text centered,above of=TL,yshift=3em] (d1) {$d$};
\node [text centered,left of=TI,xshift=10em] (d+dT){$d+d_T$};
\node [text centered,left of=C,xshift=10em] (d+dT+dH){$d+d_T+d_H$};
\node [text centered,above of=TT,yshift=3em] (d2) {$d+d_T$};
\node [text centered,left of=H,xshift=-5em] (d+dH){$d+d_H$};
\node [text centered,below of=d+dH,yshift=-2.5em] (d3){$d$};
\node [text centered,left of=HT,xshift=-5em] (d4){$d+d_H$};
\node [text centered,below of=d4,yshift=-3em] (d5){$d+d_H$};
\node [text centered,left of=C1,xshift=10em] (d6){$d+d_H$};
\node [text centered,left of=C1T,xshift=10em] (d7){$d+d_H+ \gamma d_H^T$};
\node [text centered,right of=C2T,xshift=-10em] (d8){$d+d_H+ \alpha d_H^T$};
\path (Lambda) edge (S);
\path (S) edge[dashed] (d);
\path (TL) edge[bend left=10, dashed] (d1);
\path (TI) edge[dashed] (d+dT);
\path (C) edge[dashed] (d+dT+dH);
\path (TT) edge[dashed] (d2);
\path (H) edge[dashed] (d+dH);
\path (HL) edge[bend left=5,dashed] (d3);
\path (HT) edge[dashed] (d4);
\path (C2) edge[bend left=10, dashed] (d5);
\path (C1) edge[dashed] (d6);
\path (C1T) edge[dashed] (d7);
\path (C2T) edge[dashed] (d8);
\path(S)edge node[anchor=south]{$\beta e$} (TL);
\path(S)edge node[left] {$\lambda \sigma$}(H);
\path(TL)edge node[anchor=south]{$k_1$}(TI);
\path(H)edge node[left] {$r_2$}(HT);
\path(TL)edge node[left] {$\lambda \sigma$}(HL);
\path(TI)edge node[left] {$\lambda \sigma$}(C);
\path(HL)edge node[anchor=south]{$k_2$}(C);
\path(H)edge node[anchor=south]{$\beta e$}(HL);
\path(C1)edge node[anchor=south]{$\rho_1$}(C2);
\path(C1T)edge node[anchor=north]{$\rho_1$}(C2T);
\path(C)edge node[left] {$r_3$}(C1);
\path(C1)edge node[right] {$\eta_1$}(C1T);
\path(C2)edge node[left] {$\eta_2$}(C2T);
\path(TI)edge[bend right=30] node[left,yshift=2em] {$r$}(TT);
\path(TT)edge[bend right=30] node[right,yshift=2em] {$r_1$}(S);
\path(C2)edge[bend left=15] node[right,yshift=0.5em] {$\rho_2$}(H);
\path(C2T)edge[bend left=30] node[right,yshift=0.5em] {$\rho_2$}(HT);
\end{tikzpicture}
\end{center}
\caption{Schematic Diagram}\label{fig:1}
\end{figure}
We assume that all individuals in a given compartment are identically infectious, which might ignore potential effects caused due to variation among individuals. The susceptible cannot get HIV and TB infection simultaneously that means there is no direct transmission from the class of susceptible to the class of individuals co-infected with HIV and TB. The susceptible population is increased by the constant recruitment rate $\Lambda$ which is assumed to simplify the model. All individuals in different compartments suffer from natural death rate $d$ while $d_T$ and $d_H$ are disease induced death rates due to TB and HIV separately. We ignore the temporal immunity to recover from latent TB because nowdays modern lifestyle has lowered the immunity \cite{immunity} and consider direct and endogenous reinfection only. Early initiation of ART increases the probability of developing IRIS. Since initiation of ART with TB treatment alters the disease induced death rate, therefore we consider it to be $d_H^T$ and $\gamma$ is the probability of developing IRIS during the early phase of co-treatment for compartment $C_1^T$ and $\alpha=\frac{\gamma \rho_1}{\rho_1+\rho_2}$ during late phase of co-treatment for compartment $C_2^T$ \cite{mallela}. Hence, the rate of IRIS development in $C_2^T$ compartment decreases as $1/\rho_1$ increases. We assume $\sigma$ and $e$ are per capita contact rates for HIV and TB respectively. We assume $\beta$ to be the probability of TB infection per contact with a person with active TB and $\lambda$ is the probability of HIV infection per contact with a HIV infectious person. We assume that people under treatment for any disease are not infectious for spreading that disease since they are aware of their illness, so are precautious to the spread of disease. Secondly, treatment of TB reduces its infectiousness rapidly \cite{TBinf}. In this model, we do not consider the treatment of latent TB and only sexual transmission of HIV is considered. The force of infection $\lambda_T$ associated with TB is given by
\begin{align}\label{eq:lt}
\lambda_T = \beta e \frac{(T_I + C)}{N}.
\end{align}

The force of infection $\lambda_H$ associated with HIV is given by
\begin{align}\label{eq:lh}
\lambda_H = \lambda \sigma \frac{(H + H_L + C + C_1 + C_2)}{N}.
\end{align}

 We further assume that co-infected individuals under TB-treatment die of HIV or IRIS only. We assume $k_1$ is the progression rate from latent to active TB with no HIV while $k_2$ is the progression rate from latent to active TB with HIV, $r$ is the per capita TB treatment rate with no HIV and $r_1$ is the recovery rate by treatment from TB with no HIV. The people successfully treated with TB return to class of susceptible since TB can reoccur\cite{TBt}. Let $r_2$ be the per capita HIV treatment rate with no TB while $r_3$ is the per-capita TB treatment rate in co-infected individuals. We assume $\rho_1$ and $\rho_2$ to be the transition rates of TB treatment from early phase to late phase and from late phase to completion phase respectively while $\eta_1$ and $\eta_2$ are the rates at which HIV treatment begins during early phase of TB treatment and late phase of TB treatment respectively. The assumptions result in the following differential equations that describe the interaction of the two disease model as:
 \begin{align}\label{eq:main} 
\frac{dS}{dt}&= \Lambda - \beta e S \frac{(T_I + C)}{N} - dS - \lambda \sigma S \frac{(H + H_L + C + C_1 + C_2)}{N} + r_1 T_T,\nonumber \\
\frac{dT_L}{dt}&= \beta e S \frac{(T_I + C)}{N} -(d+k_1) T_L -  \lambda \sigma T_L \frac{(H + H_L + C + C_1 + C_2)}{N},\nonumber\\
\frac{dT_I}{dt}&= k_1 T_L -(d+d_T)T_I -  \lambda \sigma T_I \frac{(H + H_L + C + C_1 + C_2)}{N}-r T_I,\nonumber \\
\frac{dH}{dt} &= \lambda \sigma S \frac{(H + H_L + C + C_1 + C_2)}{N} - r_2 H -(d+d_H)H - \beta e H \frac{(T_I+C)}{N} + \rho_2 C_2,\nonumber\\
\frac{d H_L}{dt}&=  \lambda \sigma T_L \frac{(H + H_L + C + C_1 + C_2)}{N} + \beta e H \frac{(T_I+C)}{N}-(k_2+d+d_H)H_L,\nonumber\\
\frac{dC}{dt}&= k_2 H_L +  \lambda \sigma T_I \frac{(H + H_L + C + C_1 + C_2)}{N}-(d+d_T+d_H)C - r_3C, \\
\frac{dC_1}{dt}&= r_3C-\rho_1C_1-(d+d_H)C_1-\eta_1C_1,\nonumber \\
\frac{dC_2}{dt}&=\rho_1C_1-\rho_2C_2-(d+d_H)C_2-\eta_2C_2, \nonumber\\
\frac{dC_1^T}{dt}&=\eta_1C_1-\rho_1C_1^T-(d+d_H+ \gamma d_H^T)C_1^T,\nonumber \\
\frac{dC_2^T}{dt}&=\eta_2C_2+\rho_1C_1^T-(d+ d_H + \alpha d_H^T)C_2^T,  \nonumber\\
\frac{dH^T}{dt}&=\rho_2C_2^T+r_2H-(d+d_H)H^T, \nonumber\\
\frac{dT_T}{dt}&= rT_I - r_1T_T-(d+d_T)T_T.\nonumber
\end{align}
The model flow diagram is shown in figure \ref{fig:1}. The solid arrows show the flow within the system while dashed arrows show flow out from the system.
The non-negative initial conditions are chosen.

\subsection{\textbf{Positivity and boundedness of solutions}}
Since the system deals with human population which can not be negative, we need to show that all the variables are always non-negative as well as the solutions of system \eqref{eq:main} remain positive always with positive initial conditions in the bounded region defined by
\begin{align*}
 D =\left\{(S, T_L, T_I, H, H_L, C, C_1, C_2, C_1^T, C_2^T, H^T, T_T)\in\mathbb{R}_+^{12}: N(t)\leqslant\frac{\Lambda}{d}\right\}.
\end{align*}
The next result is very obvious and follows from direct algebraic and differential calculations. 
\begin{thm}\label{th1}
For the non-negative initial conditions, the solutions $S,\; T_L,\; T_I,\; H,\; H_L,\; C,\; C_1,\\ C_2,\; C_1^T,\; C_2^T,\; H^T,\; T_T$ of the system \eqref{eq:main} are positive when $t\geqslant0$ and the region D is positively invariant.
\end{thm}

\section{\textbf{TB Sub-model}}\label{Section3}
We have the TB sub-model when $H= H_L = C = C_1 = C_2 = C_1^T = C_2^T = H^T = 0$, which is given by:
\begin{equation}\label{eq:tb}
\begin{aligned}
 \frac{dS}{dt} &= f_1 = \Lambda - \frac{\beta e S}{N}T_I - dS +r_1 T_T,\\
 \frac{dT_L}{dt} &= f_2 = \frac{\beta e S}{N}T_I - dT_L -k_1 T_L,\\
 \frac{dT_I}{dt} &= f_3 = k_1T_L - (d+d_T)T_I -rT_I,\\
\frac{dT_T}{dt} &= f_4 = rT_I - r_1T_T - dT_T - d_T T_T
\end{aligned}
\end{equation}
with non-negative initial conditions and $\lambda_T = \frac{\beta e T_I}{N}$ is the force of infection. The total population is given by
\begin{align*}
N(t) = S(t)+ T_L(t) + T_I(t) + T_T(t).
\end{align*}
Considering biological constraints, the system \eqref{eq:tb} will be studied in the following region:
\begin{align*}
D_1 =\left\{(S, T_L, T_I, T_T)\in \mathbb{R}_+^{4}: N(t)\leqslant\frac{\Lambda}{d}\right\}.
\end{align*}
It can be easily shown that the solutions $S, T_L, T_I, T_T$ of the system are bounded and positively invariant in $D_1$.
\subsection{\textbf{Disease free equilibrium and stability analysis}}\label{subsection:3.1}

The disease free equilibrium is given by
\begin{align*}
E_0^T= (S_0, T_{L_0}, T_{I_0}, T_{T_0})= \left(\frac{\Lambda}{d},0,0,0\right).
\end{align*}

The basic reproduction number is defined as the average number of new cases of an infection caused by one typical infected individual in a population consisting of susceptibles only \cite{Dkm1, Jones, Dkm2}. Here reproduction number, $\mathcal{R}_0^T$, is defined as the number of TB infections produced by an active TB case. We use the next generation matrix method to find basic reproduction number and the basic reproduction number is given by 
\begin{align}\label{eq:Rnott}
 \mathcal{R}_0^T = \frac{\beta e k_1 \Lambda}{N d(d^{2}+dr+dd_T+dk_1+rk_1+d_Tk_1)}.
\end{align}
We now discuss the stability of disease free equilibrium.
\begin{thm}\label{th2}
The disease-free equilibrium, $E_0^T$ is locally asymptotically stable when $\mathcal{R}_0^T<1$ and unstable when $\mathcal{R}_0^T>1$.
\begin{proof}
The Jacobian matrix of the system \eqref{eq:tb} at $E_0^T$ is given by 
\begin{equation}
J(E_0^T)= \begin{bmatrix}
-d & 0 & -\frac{\beta e \Lambda}{N d} & r_1\\
0 & -d-k_1 & \frac{ \beta e \Lambda }{N d}& 0\\
0 & k_1 & -(d+d_T+r) & 0\\
0 & 0 & r & -(r_1+d+d_T)
\end{bmatrix}.
\label{eq:a}
\end{equation}
The characteristic polynomial is given by
\begin{equation}\label{eq:cp1}
 (x+d)(x+d+d_T+r_1) \left(x^2+(2d+r+d_T+k_1)x 
+(d^2+dr+dd_T+dk_1+rk_1+d_Tk_1)-\frac{e\beta \Lambda k_1}{N d}\right).
\end{equation}
The first two factors are linear and give eigenvalues $-d$ and $-(d+d_T+r_1)$ for \eqref{eq:cp1}, which have negative real parts. For the remaining quadratic factor we use Routh-Hurwitz stability criterion by which  all the coefficient of quadratic polynomial are positive if $\mathcal{R}_0^T<1$. Hence, $E_0^T$ is locally asymptotically stable for $\mathcal{R}_0^T<1$ \cite{Perko}.
\end{proof} 
\end{thm}
We now list two conditions that are sufficient to guarantee the global stability of the disease free equilibrium point. Following Castillo-Chavez et al. \cite{CZ}, we rewrite the model system \eqref{eq:tb} as
\begin{equation}\label{equation:3.5}
\begin{aligned}
&\frac{d\textbf{S}}{dt}= F(\textbf{S, I}),\\
&\frac{d\textbf{I}}{dt} = G(\textbf{S, I}),
\end{aligned}
\qquad
\begin{aligned}
 G(\textbf{S, 0})=0,
\end{aligned}
\end{equation} 
where $\textbf{S}\in\mathbb{R}^2$ denotes the number of uninfected individuals and $\textbf{I}\in\mathbb{R}^2$ denotes the number of infected individuals. 
$\textbf{E}_0=(X^*, 0)$ denotes the disease free equilibrium of system \eqref{equation:3.5}. The conditions $(H1)$ and $(H2)$ below must be satisfied to guarantee a local asymptotic stability.
\newline
(H1) For $ \frac{d\textbf{S}}{dt}=F(\textbf{S},0), S^*$ is globally asymptotically stable, \newline
(H2) $G(\textbf{S, I}) = A\textbf{I}-\hat{G}(\textbf{S,  I}),\;\hat{G}(\textbf{S, I})\geq0$ for $(\textbf{S, I})\in G_2$,\\ 
where $A= D_{\textbf{I}}G(S^*, 0)$ is an M-matrix (the off diagonal elements are non negative) and $G$ is the region where the model makes biological sense. If the system \eqref{eq:tb} satisfies the above two conditions, then the theorem holds.
\begin{thm}\label{thm3}
The fixed point $E_0^T=(S^*,\, 0)$ is a globally asymptotically stable equilibrium of system \eqref{eq:tb} if $\mathcal{R}_0^T<1$ and the assumptions in \eqref{equation:3.5} are satisfied.
\begin{proof}
 We have already proved in theorem \eqref{th2} that for $\mathcal{R}_0^T<1$, $E_0^T$ is locally asymptotically stable. Consider
\begin{flalign}
&\frac{dS}{dt}=\Lambda- \frac{\beta e S}{N}T_I - dS +r_1 T_T = F(S,I),\nonumber \\
&\frac{dI}{dt}= G(S, I),
\end{flalign}

where $G(S,I)= \begin{bmatrix}
\frac{\beta e S T_I}{N}-d T_L -k_1T_L \\
k_1 T_L-(d+d_T)T_I-rT_I\\
rT_I-r_1T_T-dT_T-d_TT_T
\end{bmatrix}$ and 
$F(S, 0)= \begin{bmatrix}
\Lambda - dS\\
0
\end{bmatrix}.$ \newline
Since $dS/dt=F(S,0)$ is a linear equation, $S^*$ is globally stable. Hence, (H1) holds.\newline  
$G(S, I)= AI- \hat{G}(S, I)$ and $
A = \begin{bmatrix}
-d-k_1 & \beta e & 0\\
k_1 & -d-d_T-r & 0\\
0 & r & -r_1-d-d_T
\end{bmatrix}. $ \newline
We get, \newline
$\hat{G}(S, I) = \begin{bmatrix}
\hat{G_1}(S, I)\\
\hat{G_2}(S, I)\\
\hat{G_3}(S, I)
\end{bmatrix}
= \begin{bmatrix}
\beta e T_I(1-\frac{S}{N})\\
0\\
0
\end{bmatrix}.
$ \newline
Since $S$ is always less than or equal to N, $\frac{S}{N} \leq1$ and $\hat{G_1}(S, I)\geqslant1$.
Thus, $\hat{G}(S,I)\geq0$ and this implies that $E_0^T$ is globally asymptotically stable.
\end{proof}
\end{thm} 
\subsection{\textbf{Existence and Stability analysis for endemic equilibrium point}}\label{subsection:3.2}
To find conditions for the existence of an equilibrium for which TB is endemic in the population, denoted by $E_1^T=(\hat{S}, \hat{T_L}, \hat{T_I}, \hat{T_T})$, the equations in \eqref{eq:tb} are solved in terms of force of infection at steady state $\lambda^*_T$, given by 
\begin{align}
\lambda_T^*=\frac{\beta e T_I}{N}.
\end{align}
The system \eqref{eq:tb} reduces to the following:
\begin{align*}
\begin{bmatrix}
-(\lambda_T^*+d) &0 & 0 & r_1\\
\lambda_T^* & -(d+k_1) & 0 & 0\\
0 & k_1 & -d-d_T-r &0\\
0 &0 & r & -(r_1+d+d_T) 
\end{bmatrix}
\begin{bmatrix}
S\\
T_L\\
T_I\\
T_T
\end{bmatrix}=\begin{bmatrix}
-\Lambda\\
0\\
0\\
0
\end{bmatrix}.
\end{align*}
Solving this, we get
\begin{equation}\label{eq:b}
\begin{aligned}
&\hat{S}= \frac{\Lambda(d+r+d_T)(d+k_1)(d+d_T+r_1)}{A}, \\
&\hat{T_L}=\frac{\Lambda(\lambda^*_T(d+r+d_T))(d+d_T+r_1)}{A},\\
&\hat{T_I}=\frac{\Lambda\lambda^*_T k_1(d+d_T+r_1)}{A},\\
&\hat{T_T}=\frac{r\Lambda\lambda^*_T k_1}{A},\\
\end{aligned}
\end{equation}
where A = $(d+r+d_T)(d+\lambda^*_T)(d+k_1)(d+d_T+r_1)-r\lambda_T^* k_1 r_1$. Using (3.8) in (3.7), we get
\begin{align*}
\lambda^*_T=\frac{\beta e \Lambda k_1 \lambda^*_T(d+d_T+r_1)} {N(d+r+d_T)(d+\lambda^*_T)(d+k_1)(d+d_T+r_1)-r \lambda^*_T k_1 r_1}, \\
\end{align*}
which reduces to
\begin{align*}
 \frac{\lambda^*_T {\left[\frac{\beta e \Lambda k_1(d+r_1+d_T)}{N}-(d+r+d_T)(d+\lambda^*_T)(d+k_1)(d+d_T+r_1)+r \lambda^*_T k_1 r_1\right]}}{(d+r+d_T)(d+\lambda^*_T)(d+k_1)(d+d_T+r_1)-r \lambda^*_T k_1 r_1}=0, 
\end{align*}
where $\lambda^*_T = 0$ corresponds to the disease free equilibrium and 
\begin{equation}\label{eq:lmdt}
 \lambda^*_T=\frac{\frac{\beta e \Lambda k_1(d+d_T+r_1)}{N}-d(d+r+d_T)(d+k_1)(d+d_T+r_1)}{(d+d_T+r)(d+k_1)(d+d_T+r_1)-r k_1 r_1},
 \end{equation}
  corresponds to the existence of endemic equilibrium. For a disease to spread, the force of infection $(\lambda^*_T)$ should be positive. It can be seen clearly that denominator of (3.9) is always positive. So for $\lambda^*_T$ to be positive, its numerator should be positive. Therefore,
\begin{align}
& \frac{\beta e \Lambda k_1}{N}-d(d+r+d_T)(d+k_1)(d+d_T+r_1)>0,\nonumber\\
& \frac{\beta e \Lambda k_1}{N d (d+r+d_T)(d+k_1)}>1,\nonumber\\
 \Rightarrow & \mathcal{R}_0^T> 1.
\end{align}
Thus, (3.9) reduces to $\mathcal{R}_0^T>1$. We have just proved the following result.
\begin{thm}\label{lem:thm3}
The submodel system (3.1) has a unique endemic equilibrium whenever $\mathcal{R}_0^T>1$.
\end{thm}
\subsection{\textbf{Local stability of endemic equilibrium}}
We prove the local asymptotic stability of the endemic equilibrium $E_1^T$, using the center manifold theory, as described in \cite[Theorem 4.1]{CB1}, with $E_1^T=(\hat{S}, \hat{T_L}, \hat{T_I}, \hat{T_T})$ and each of its components as given in (3.7). Firstly we simplify the system \eqref{eq:tb} to apply this method. Let $S=x_1, T_L=x_2, T_I=x_3$ and $ T_T=x_4$, so that $N=x_1+x_2+x_3+x_4$, then the system can be written in the form $\frac{dX}{dt}=(f_1, f_2, f_3, f_4)^T$.
The basic reproduction number of the system \eqref{eq:tb} is given by \eqref{eq:Rnott}. We choose a bifurcation parameter $\beta^*$, by solving for $\beta e$ from $\mathcal{R}_0^T=1$:
\begin{align*}
\beta^*= \frac{N d(d+r+d_T)(d+k_1)}{\Lambda k_1}.
\end{align*}  
The system \eqref{eq:tb} has a disease free equilibrium given by 
\begin{align*}
E_0^T=(x_{1_0}, x_{2_0}, x_{3_0}, x_{4_0})= \left( \frac{\Lambda}{d}, 0, 0, 0\right).
\end{align*}
The Jacobian matrix of the linearized system of (3.11) evaluated at $E_0^T$, $J(E_0^T)$ is given as in \eqref{eq:a}. $J(E_0^T)|_{\beta^*}$ has a zero eigenvalue which is simple and all the other eigenvalues have negative real parts, therefore the center manifold theory can be applied. \newline
The Jacobian matrix $J(E_0^T)|_{\beta^*}$ has a right eigenvector, associated with zero eigenvalue which is given by $w=(w_1, w_2, w_3, w_4)^T$, where
\begin{equation}
\begin{aligned}
& w_1=\frac{r_1}{d}-\frac{(d+r+d_T)(d+k_1)(d+d_T+r_1)}{d r k_1}, \\
& w_2 = \frac{(d+r+d_T)(d+d_T+r_1)}{r k_1},\\
& w_3 = \frac{(d+d_T+r_1)}{r},\\
& w_4 = 1.
\end{aligned} 
 \end{equation}
 Since $x_{1_0}>0$, there is no restriction on sign of $w_1$ [refer \cite{CB1}, Remark 1, pg. 375], while $w_i>0\;\forall\;i\neq1$.
 \newline
 Further, $J(E_0^T)|_{\beta^*}$ has a left eigenvector $v=(v_1, v_2, v_3, v_4)$, associated with the zero eigenvalue, where
 \begin{equation}
 \begin{aligned}
& v_1 = 0, \\
& v_2 = \frac{k_1}{(d+k_1)}, \\
& v_3 = 1, \\
& v_4 = 0. 
\end{aligned}
\end{equation}
The local stability near the bifurcation point $\beta^*=\beta e$ is determined by the signs of the two associated constants, denoted by $a$ and $b$ which are defined by
\begin{equation}\label{eq:abl}
\begin{aligned}
& a= \sum\limits_{k,i,j=1}^4 v_k w_i w_j\frac{\partial^2{f_k}}{\partial x_i \partial x_j}(0, 0),\\
& b = \sum\limits_{k,i,j=1}^4 v_k w_i \frac{\partial ^2 f_k}{\partial x_i \partial \phi}(0, 0), 
\end{aligned}
\end{equation}
with $\phi =\beta e - \beta^*$ and for $\beta e = \beta^*, \phi=0$. \newline
For the system \eqref{eq:tb}, the associated non-zero partial derivatives at $E_0^T$ are
\begin{equation*}
\begin{aligned}
\frac{\partial^2f_2}{\partial x_2 \partial x_3}=& -\frac{\beta^* d}{\Lambda},\\ 
\frac{\partial^2f_2}{\partial x_3 \partial x_4}=& -\frac{\beta^* d}{\Lambda},
\end{aligned}
\qquad
\begin{aligned}
\frac{\partial^2 f_2}{\partial x_3^2}=& -2 \frac{ \beta^* d}{\Lambda},\\
\frac{\partial^2 f_2}{\partial x_3 \partial \beta^*}=&1.
\end{aligned}
\end{equation*}
From the above expressions, we get
\begin{equation}
\begin{aligned}
 a  = & -\frac{(d+d_T+r)^2(d+r+d_T)^2 d}{\Lambda^2 r k_1} -\frac{2N d^2 (d+d_T+r_1)^2 (d+r+d_T)}{r^2 \Lambda^2}\\
 & -\frac{N d^2(d+d_T+r_1)(d+d_T+r)}{r \Lambda^2},\\
b  = & \frac{k_1 (d+d_T+r_1)}{r(d+k_1)}.
\end{aligned}
\end{equation}
Since all the terms in the expressions in (3.15) are positive, therefore $a<0$ and $b>0$. Thus, \cite[Theorem 4.1]{CB1} implies that the unique equilibrium point of  system \eqref{eq:tb}, which exists when $\mathcal{R}_0^T>1$ is locally asymptotically stable when $\beta^*< \beta e$ with $\beta e$ close to $\beta^*$ and the system undergoes transcritical bifurcation at $\beta^*=\beta e$ which is supercritical or forward. Hence, the following result is established.
\begin{thm}
The endemic equilibrium $E_1^T$ is locally asymptotically stable for the basic reproduction number $\mathcal{R}_0^T>1$ and the system undergoes  supercritical transcritical bifurcation at $R_0=1$ and $\beta^*=\beta e$ acts as the bifurcation parameter.
\end{thm}   

\section{\textbf{HIV Submodel}}\label{Section4} 
We have HIV submodel when $T_L= T_I = T_T = C = C_1 = C_2 = C_1^T = C_2^T = 0$, which is given by the following system of equations:
\begin{equation}\label{eq:hiv}
\begin{aligned}
 \frac{dS}{dt}= & \Lambda -dS - \frac{\lambda \sigma S H}{N}, \\
 \frac{dH}{dt}= & \frac{\lambda \sigma S H}{N}-r_2 H-(d+d_H)H, \\
 \frac{dH^T}{dt}= & r_2 H -(d+d_H)H^T,
\end{aligned}
\end{equation}
with non-negative initial conditions and the force of infection is given by
\begin{align*}
\lambda_H= \frac{\lambda \sigma H}{N},    
\end{align*}
and the total population for the system is $N= S+H+H^T$. Due to biological constraints, the system \eqref{eq:hiv} is studied in the following region
\begin{align*}
D_2 =\left\{(S,\; H,\; H^T)\in \mathbb{R}_+^{3}: N(t)\leqslant\frac{\Lambda}{d}\right\}.
\end{align*}
This is easy to prove that the solutions $S,\; H$ and $H^T$ of the system \eqref{eq:hiv} are bounded and positively invariant in $D_2$.

\subsection{\textbf{Disease free equilibrium and stability analysis}}
The disease free equilibrium for the system \eqref{eq:hiv} is given by
\begin{align*}
E_0^H=(S_0,\; H_0,\; H^T_0)= \left(\frac{\Lambda}{d},\; 0,\; 0,\; 0\right),
\end{align*}
and the basic reproduction number is given by
\begin{equation}
\mathcal{R}_0^H= \frac{\Lambda \lambda \sigma}{N d(r_2+d+d_H)}.
\end{equation}
Next we discuss the stability of disease free equilibrium.
\begin{thm}\label{theorem:thm4.1}
The disease free equilibrium, $E_0^H$ is locally asymptotically stable when $\mathcal{R}_0^H<1$ and unstable when $\mathcal{R}_0^H>1$.
\begin{proof}
The Jacobian matrix of the system \eqref{eq:hiv} at $E_0^H$ is given by
\begin{equation}
 J(E_0^H)=\begin{bmatrix}
-d & -\frac{\Lambda\lambda \sigma}{N d} & 0\\
0 & \frac{\Lambda\lambda \sigma}{N d}-(r_2+d+d_H) & 0\\
0 & r_2 & -(d+d_H)
\end{bmatrix},
\end{equation}
which has characteristic polynomial
\begin{align}
(d + x) ( x + d + d_H) \left(\frac{\Lambda \lambda \sigma}{N d} -x -d - d_H - r_2\right).
\end{align}
Eigenvalues for (4.3) are $-d, -(d+d_H)$ and $ (\frac{\Lambda\lambda \sigma}{N d}-(d+d_H+r_2) )$ which have negative real parts when $\frac{\Lambda\lambda \sigma}{N d(d+d_H+r_2)}<1$ i.e. $\mathcal{R}_0^H<1$. Thus, $E_0^H$ is locally asymptotically stable for $\mathcal{R}_0^H<1$ and unstable for $\mathcal{R}_0^H>1$.
\end{proof}
\end{thm}
We now discuss the global stability of disease free equilibrium $E_0^H$.
\begin{thm}
The fixed point $E_0^H=(S^*,0)$ is a globally asymptotically stable equilibrium of system \eqref{eq:hiv} provided that $\mathcal{R}_0^H<1$ and the assumptions in \eqref{equation:3.5} are satisfied.
\begin{proof}
Proof is similar to the theorem \ref{thm3} and hence omitted.
\end{proof}
\end{thm}

\subsection{\textbf{Existence and Stability of Endemic Equilibrium point}}
The endemic equilibrium is given by
\begin{align*}
E_1^H= (\tilde{S},\; \tilde{H},\; \tilde{H}^T),
\end{align*}
where
\begin{equation}\label{eq:4.6}
\begin{aligned}
 \tilde{S}=&\frac{N(r_2+d+d_H)}{\lambda \sigma},\\
\tilde{H}=& \frac{\Lambda}{(r_2+d+d_H)}- \frac{d N}{\lambda \sigma},\\
 \tilde{H^T}= & \frac{r_2}{(r_2+d+d_H)}\left( \frac{\Lambda}{(r_2+d+d_H)}-\frac{d N}{\lambda \sigma}\right).
\end{aligned}
\end{equation}
The HIV endemic exists when $\lambda_H^*$, given by
\begin{align}
\lambda_H^* = \frac{\lambda \sigma \hat{H}}{N},
\end{align}
is positive. Using (4.6) in (4.7), if $\lambda_H^* > 0$, then
\begin{align*}
& \frac{\lambda \sigma}{N} \left( \frac{\Lambda}{(r_2 +d +d_H)}- \frac{d N}{\lambda \sigma} \right) > 0, \\
& \frac{\lambda \sigma}{(r_2 +d +d_H)} > \frac{N d}{\Lambda},\\ 
\end{align*}
Therefore, $\mathcal{R}_0^H > 1.$
Thus, the endemic equilibrium exists when $\mathcal{R}_0^H>1$. We have just proved the following result.
\begin{thm}\label{thm:4.3}
The endemic equilibrium, $E_1^H$ exists whenever $\mathcal{R}_0^H>1.$  
\end{thm}
We now discuss the stability of the endemic equilibrium point $E_1^H$. 
\begin{thm}
The endemic equilibrium $E_1^H$ is locally asymptotically stable for the basic reproduction number $\mathcal{R}_0^H>1$.
\end{thm}
\begin{proof}
The Jacobian matrix of the system \eqref{eq:hiv} at $E_1^H$ is given by
\begin{equation}
\begin{bmatrix}
\frac{\Lambda \lambda \sigma}{N(r_2+d+d_H)} & -(r_2+d+d_H) & 0\\
\frac{\Lambda \lambda \sigma}{N(r_2+d+d_H)}-d & 0 & 0\\
0 & r_2 & -(d+d_H)
\end{bmatrix}.
\end{equation} 
The characteristic equation for (4.8) is given by
\begin{equation}
(x+d+d_H)\left(x^2+\frac{\lambda \sigma x}{N(d+d_H+r_2)}- d(r_2+d+d_H)+\frac{\Lambda \lambda \sigma}{N}\right)= 0.
\end{equation}
The factor $(x+d+d_H)$ gives an eigenvalue $-d-d_H$, which has negative real part. For the other quadratic factor we use Routh Hurwitz criterion of stability, by which all the coefficients in the quadratic factor should be positive when $ \mathcal{R}_0^H > 1$. Therefore, in $D_2$ whenever  $\mathcal{R}_0^H>1$, $E_1^H$ is locally asymptotically stable.
\end{proof}
\section{\textbf{Analysis of the main model}}\label{Section5}
 In this section, we analyze the main model \eqref{eq:main}. Biologically, the full model can have four equilibria, namely, disease free equilibrium $E_0$, TB only endemic equilibrium $E^T$, HIV only endemic equilibrium $E^H$ and the interior endemic equilibrium point $E^{TH}$. 
 \subsection{\textbf{The disease free equilibrium and stability analysis}}
 The disease free equilibrium is given by
\begin{align*}
 E_0 = \left( \frac{\Lambda}{d}, 0, 0, 0, 0, 0, 0, 0, 0, 0, 0, 0 \right).
\end{align*}
First, we calculate the basic reproduction number by next generation operator method as in subsection \eqref{subsection:3.1}. The transition matrix $T$ and the transmission matrix $\Sigma$, are as follows:
\begin{align}
T=\begin{bmatrix}
0 & \frac{\beta e \Lambda}{N d} & 0 & 0 & \frac{\beta e \Lambda}{N d} & 0 & 0\\
0 & 0 & 0 & 0 & 0 & 0 & 0\\
0 & 0 & \frac{\Lambda \lambda \sigma}{N d} & \frac{\Lambda \lambda \sigma}{N d} & \frac{\Lambda \lambda \sigma}{N d} & \frac{\Lambda \lambda \sigma}{N d} & \frac{\Lambda \lambda \sigma}{N d}\\
0 & 0 & 0 & 0 & 0 & 0 & 0\\
0 & 0 & 0 & 0 & 0 & 0 & 0\\
0 & 0 & 0 & 0 & 0 & 0 & 0\\
0 & 0 & 0 & 0 & 0 & 0 & 0\\
\end{bmatrix},
\end{align}
\begin{align}
\Sigma=\begin{bmatrix}
-d-k_1 & 0 & 0 & 0 & 0 & 0 & 0 \\
k_1 & -d-d_T-r & 0 & 0 & 0 & 0 & 0\\
0 & 0 & -r_2-d-d_H & 0 & 0 & 0 & \rho_2\\
0 & 0 & 0 & k_2-d-d_H & 0 & 0 & 0\\
0 & 0 & 0 & k_2 & C_1 & 0 & 0\\
0 & 0 & 0 & 0 & r_3 & C_2 & 0\\
0 & 0 & 0 & 0 & 0 & \rho_1 & C_3
\end{bmatrix},
\end{align}
where \\
$C_1= -d-d_T-d_H-r_3,\\
C_2=-d-d_H-\rho_1-\eta_1,\\
C_3=-d-d_H-\rho_2-\eta_2.$ \newline
The dominant eigenvalues of $-T\Sigma^{-1}$ are 
\begin{align*}
\mathcal{R}_0^T =& \frac{\beta e \Lambda k_1}{N d(d^2 + dr + dd_T+ dk_1 + rk_1 + d_T k_1)} = \frac{\beta e \Lambda k_1}{N d (d+r+d_T)(d+k_1)},\\
 \mathcal{R}_0^H =& \frac{\Lambda \lambda \sigma}{N d(r_2+d+d_H)}. 
\end{align*}
Thus, the basic reproduction number of the model \eqref{eq:main} is given by
\begin{align}\label{rp}
\mathcal{R}_0 = \max\{\mathcal{R}_0^T, \mathcal{R}_0^H\}.
\end{align}
\begin{thm}
The disease free equilibrium, $E_0$ is locally asymptotically stable when $\mathcal{R}_0<1$ and unstable when $\mathcal{R}_0>1$.
\begin{proof}
The Jacobian matrix $J(E_0)$ of the model system \eqref{eq:main} at $E_0$ is given by
\begin{equation}\label{eq:jmain}
\begin{bmatrix}
-d & 0 & -\frac{\beta e \Lambda}{N d} & -\frac{\Lambda \lambda \sigma}{N d} & -\frac{\Lambda \lambda \sigma}{N d} & \frac{-\Lambda(\beta e + \lambda \sigma)}{N d} & -\frac{\Lambda \lambda \sigma}{N d} & -\frac{\Lambda \lambda \sigma}{N d} & 0 & 0 & 0 & r_1 \\ 
0 & -d-k_1 & \frac{ \beta e \Lambda}{N d} & 0 & 0 & 0  & 0 & 0 & 0 & 0 & 0 & 0 \\
0 & k_1 & C_0 & 0 & 0 & 0 & 0 & 0 & 0 & 0 & 0 & 0\\
0 & 0 & 0 & C_1 & \frac{\Lambda \lambda \sigma}{N d} & \frac{\Lambda \lambda \sigma}{N d} & \frac{\Lambda \lambda \sigma}{N d} & C_2 & 0 & 0 & 0 & 0\\
0 & 0 & 0 & 0 & C_3 & 0 & 0 & 0 & 0 & 0 & 0 & 0\\
0 & 0 & 0 & 0 & k_2 & C_4 & 0 & 0 & 0 & 0 & 0 & 0\\
0 & 0 & 0 & 0 & 0 & r_3 & C_5 & 0 & 0 & 0 & 0 & 0\\
0 & 0 & 0 & 0 & 0 & 0 & \rho_1 & C_6 & 0 & 0 & 0 & 0\\
0 & 0 & 0 & 0 & 0 & 0 & \eta_1 & 0 & C_7 & 0 & 0 & 0\\
0 & 0 & 0 & 0 & 0 & 0 & 0 & \eta_2 & \rho_1 & C_8 & 0 & 0\\
0 & 0 & 0 & 0 & 0 & 0 & 0 & 0 & 0 & \rho_2 & -d-d_H & 0 \\
0 & 0 & r & 0 & 0 & 0 & 0 & 0 & 0 & 0 & 0 & -d-d_T     
\end{bmatrix}
\end{equation} 
where, \\
$C_0 =-(d+r+d_T), \\
C_1=  \frac{\Lambda \lambda \sigma}{N d}-(d+d_H+r_2),\\
C_2= \frac{\Lambda \lambda \sigma}{N d}+\rho_2,\\
C_3=-(d+d_H+k_2),\\
C_4=-(d+d_H+d_T+r_3),\\
C_5=-(d+d_H+\eta_1+\rho_1),\\
C_6=-(d+d_H+\eta_2+\rho_2),\\
C_7=-(d+ d_H+ \gamma d_H^T+\rho_1),\\
C_8=-(d+ d_H+ \alpha d_H^T+\rho_2).$ 
\newline
The characteristic equation of \eqref{eq:jmain} is given by the following:
\begin{equation}
\begin{aligned}
& (x+d+d_H+k_2)(x+d+d_T+r_1)(x+d+d_T+d_H+r_3)(x+d+ d_H + \gamma d_H^T + \rho_1)\\
& (x+d)(x+d+d_H)(x+d+ d_H+\rho_1+\eta_1)(x+d+ d_H+ \alpha d_H^T+\rho_2)(x+d+d_H+\rho_2+\eta_2)\\
& \bigg(x+d+d_H+r_2-\frac{\Lambda \lambda \sigma}{N d}\bigg) \bigg(x^2+ (2 d + r + d_T + k_1)x+d^2+dr+dd_T+dk_1+rk_1\\
& + d_Tk_1-\frac{\beta e \Lambda k_1}{Nd}\bigg)=0.
\end{aligned}
\end{equation}
Clearly, the first nine factors in (5.5) give eigenvalues with negative real parts. Eigenvalue of last two factors would have negative real parts if $\mathcal{R}_0^H<1$ and $\mathcal{R}_0^T<1$ respectively by using Routh-Hurwitz Stability criterion.\\
Since all the coefficients of the quadratic factors are positive, by Routh-Hurwitz criterion the disease free equilibrium is locally asymptotically stable for $\mathcal{R}_0<1$ and unstable for $\mathcal{R}_0>1$.
\end{proof}
\end{thm} 
\subsection{\textbf{The endemic equilibria and their stability }} 
In this section, we discuss the various endemic equilibria and their stability. Biologically, there can be three endemic equilibria that are TB endemic, HIV endemic and an equilibrium point where both the diseases are endemic.
\subsubsection{\textbf{TB endemic equilibrium and stability}} 
The TB endemic is given by
\begin{equation*}
 E_T=(\hat{S},\; \hat{T_L},\; \hat{T_I},\; 0,\; 0,\; 0,\; 0,\; 0,\; 0,\; 0,\; 0,\; \hat{T_T}).
\end{equation*}
 One condition for existence of TB endemic equilibrium can be shown as in \eqref{subsection:3.2} because for TB to be endemic $\lambda_T^*>0$ which is given by \eqref{eq:lmdt} and  results in $R_0^T>1$ and the second condition will be proved by using center manifold theory. The $\hat{S},\; \hat{T_L},\; \hat{T_I}$ and $\hat{T_T}$ are given as in \eqref{eq:b}.

We prove the local stability of the endemic equilibrium $E_T$, using the center manifold theory, as described in \cite[Theorem 4.1]{CB1}. We simplify the system \eqref{eq:main} to apply this method. Let $S=x_1,\; T_L= x_2,\; T_I=x_3,\; H=x_4,\; H_L= x_5,\; C=x_6,\; C_1= x_7,\; C_2 = x_8,\; C_1^T=x_9,\; C_2^T=x_{10},\; H^T=x_{11}$ and $T_T=x_{12}$, so that $N=x_1+x_2+x_3+x_4+x_5+x_6+x_7+x_8+x_9+x_{10}+x_{11}+x_{12}$. The model system \eqref{eq:main} can be written in the form $\frac{dX}{dt}=(f_1,\; f_2,\; f_3,\; f_4,\; f_5,\; f_6,\; f_7,\; f_8,\; f_9,\; f_{10},\; f_{11},\; f_{12})^T$ maintaining the sequence.

The basic reproduction number of the system \eqref{eq:main} is given by \ref{rp}. Now, we choose a bifurcation parameter $\beta^*$, by solving for $\mathcal{R}_0^T=1$, we get,
\begin{align*}
\beta^*=\frac{N d (d+r+d_T)(d+k_1)}{\Lambda k_1}.
\end{align*}
The Jacobian matrix of the linearized system of (5.7) evaluated at disease free equilibrium point $E_0$ of system \eqref{eq:main} denoted by $J(E_0)$ and evaluated at $\beta^*$, that is, $J(E_0)|_{\beta^*}$ has a zero eigenvalue which is simple and all other eigenvalues have negative real parts when $\mathcal{R}_0^H<1$.
Therefore, we can apply center manifold theory here. \newline
The Jacobian matrix $J(E_0)|_{\beta^*}$ has a right eigenvector, associated with zero eigenvalue given by $w=(w_1,\; w_2,\; w_3,\; w_4,\; w_5,\; w_6,\; w_7,\; w_8,\; w_9,\; w_{10},\; w_{11},\; w_{12})^T$, where $w_i=0$ for all $i$ except $i=1, 2, 3$ and $12$ which are as follows:
\begin{equation}
\begin{aligned}
 w_1=& \frac{r_1}{d}-\frac{(d+r+d_T)(d+k_1)(d+d_T+r_1)}{d r k_1}, \\
 w_2 = &\frac{(d+r+d_T)(d+d_T+r_1)}{r k_1},\\
 w_3 = & \frac{(d+d_T+r_1)}{r},\\
 w_{12} = & 1.
\end{aligned}
\end{equation}
The $J(E_0)|_{\beta^*}$ has a left eigenvector $v= (v_1,\; v_2,\; v_3,\; v_4,\; v_5,\; v_6,\; v_7,\; v_8,\; v_9,\; v_{10},\; v_{11},\; v_{12})$ associated with the zero eigenvalue, where
\begin{equation}
\begin{aligned}
& v_i = 0 \; \forall \; i \neq 2,3 ,\\
& v_2 = \frac{k_1}{(d+k_1)}, \\
& v_3 = 1.
\end{aligned}
\end{equation}
For determining the local stability near the bifurcation point $\beta^*=\beta e$, we need to determine the signs of the two associated constants, $a$ and $b$, defined by \eqref{eq:abl} with $\phi =\beta e - \beta^*$ and for $\beta e = \beta^*, \phi=0$. \newline
Hence, the associated non-zero partial derivatives at $E_0$ are
\begin{equation*}
\begin{aligned}[c]
\frac{\partial^2f_2}{\partial x_2 \partial x_3}&= -\frac{\beta^* d}{\Lambda},\\
\frac{\partial^2f_2}{\partial x_3 \partial x_{12}}&= -\frac{\beta^* d}{\Lambda},
\end{aligned}
\qquad
\begin{aligned}[c]
\frac{\partial^2 f_2}{\partial x_3^2}&= -2 \frac{ \beta^* d}{\Lambda},\\
\frac{\partial^2 f_2}{\partial x_3 \partial \beta^*}&=1.
\end{aligned}
\end{equation*}
From the above calculations, we get
\begin{equation}
\begin{aligned}\label{eq:e}
 a = & -\frac{(d+d_T+r)^2(d+r+d_T)^2 d}{\Lambda^2 r k_1} -\frac{2N d^2 (d+d_T+r_1)^2 (d+r+d_T)}{r^2 \Lambda^2} \\
& - \frac{N d^2(d+d_T+r_1)(d+d_T+r)}{r \Lambda^2},\\
 b = & \frac{k_1 (d+d_T+r_1)}{r(d+k_1)}.
\end{aligned}
\end{equation}

We conclude that $a<0$ and $b>0$. Thus, our calculations together with  \cite[Theorem 4.1]{CB1} implies that there exists a TB endemic equilibrium point of system \eqref{eq:main} when $\mathcal{R}_0^T>1$ and $\mathcal{R}_0^H<1$ and is locally asymptotically stable when $\beta^*<\beta e$ with $\beta e$ close to $\beta^*$ and system undergoes supercritical transcritical bifurcation. Hence, we get the next result.
\begin{thm}
The endemic equilibrium point $E_T$ exists for $\mathcal{R}_0^T>1$ and $\mathcal{R}_0^H<1$ and is locally asymptotically stable for $\mathcal{R}_0^T$ near 1 and system undergoes supercritical transcritical bifurcation at $\mathcal{R}_0^T=1$ whereas $\beta^*$ is the bifurcation parameter.
\end{thm}
\subsubsection{\textbf{HIV endemic and its stability}}
The HIV endemic is given by
\begin{align*}
E_H=(\tilde{S},\; 0,\; 0,\; \tilde{H},\; \tilde{H}_L,\; 0,\; 0,\; 0,\; 0,\; 0,\; \tilde{H}^T,\; 0).
\end{align*}
One existence condition of HIV endemic equilibrium can be shown as in \eqref{thm:4.3} and the other condition will be proved by using center manifold theory. The $\tilde{S}, \tilde{H}, \tilde{H}_L$ and $\tilde{H}^T$ are given as in \eqref{eq:4.6}. Again we use the center manifold theory, as described in [refer \cite[Theorem 4.1]{CB1}. We choose a bifurcation parameter $\lambda^*$. By solving for $\mathcal{R}_0^T=1$, we get,
\begin{align*}
\lambda^*= \lambda \sigma =\frac{N d (d+d_H+r_2)}{\Lambda}.
\end{align*}
The Jacobian matrix $J(E_0)$ evaluated at $\lambda^*$ i.e. $J(E_0)|_{\lambda^*}$ has a simple zero eigenvalue for $\mathcal{R}_0^T<1$. Hence, center manifold theory can be applied here. Proceeding as in previous theorem  
we can easily calculate
\begin{equation}\label{eq:vab}
\begin{aligned}
& a = - \frac{(d+d_H+\eta_2+\rho_2)\lambda \sigma d(d+d_H)}{r_2 \Lambda(d+d_H+r_2+\rho_2)}\left(\frac{2(d+d_H)}{r_2}+1 \right),\\
& b = \frac{(d+d_H+\eta_2+\rho_2)(d+d_H)}{(d+d_H+r_2+\rho_2) r_2}.
\end{aligned}
\end{equation}
From \eqref{eq:vab}, we get $a<0$ and $b>0$. Thus, from our calculation and \cite[Theorem 4.1]{CB1}, there exists an HIV endemic $E_H$ of the model system \eqref{eq:main}, when $\mathcal{R}_0^T<1$ and $\mathcal{R}_0^H>1$ and is locally asymptotically stable when $\lambda^*< \lambda \sigma$ with $\lambda \sigma$ close to $\lambda^*$ and supercritical transcritical bifurcation occurs at $\lambda^*= \lambda \sigma$. Thus, we get the following result:
\begin{thm}
The endemic equilibrium point $E_H$ exists for $\mathcal{R}_0^T<1$ and $\mathcal{R}_0^H>1$ which is locally asymptotically stable for $\mathcal{R}_0^H$ near $1$ and supercritical transcritical bifurcation occurs at $\mathcal{R}_0^H=1$ and $\lambda^*$ acts as the bifurcation parameter.
\end{thm}
\subsubsection{\textbf{Interior endemic equilibrium}}

The interior equilibrium point of system \eqref{eq:main} exists when both the diseases are present in the population. For both the diseases to be endemic, force of infection $\lambda_T$ and $\lambda_H$ should be positive and given by \eqref{eq:lt} and \eqref{eq:lh} respectively.
It is given by $E_T^H=(\grave{S},\; \grave{T}_L,\; \grave{T}_I,\; \grave{H},\; \grave{H_L},\; \grave{C},\; \grave{C}_1,\; \grave{C}_2,\; \grave{C}_1^T,\; \grave{C}_2^T,\; \grave{H}^T,\; \grave{T}_T)$. 

\subsubsection{\textbf{Summary of the equilibrium points}} 
Table \ref{table:1} summarizes the existence and stability conditions on the different equilibrium points of the model system \eqref{eq:main}.
\begin{table}[ht]
\caption{Equilibrium Points and Their Existence and Stability Conditions}
\centering
\begin{tabular}{l c c}
\hline \hline
Equilibrium Point & Existence Conditions & Stability \\ [0.5ex]
\hline 
Disease free equilibrium & Always exists & l.a.s for $R_0<1$\\
TB endemic equilibrium & $R_0^T>1$ and $R_0^H<1$ & l.a.s when $R_0^T>1$ and $R_0^H<1$\\
HIV endemic equilibrium & $R_0^T<1$ and $R_0^H>1$ & l.a.s when $R_0^T<1$ and $R_0^H>1$\\
Interior equilibrium & $R_0^T>1$ and $R_0^H>1$ \\
\hline
\end{tabular}
\label{table:1}
\end{table}

\section{\textbf{Numerical results and discussion}}\label{Section6}
In the present section, numerical simulations are carried out using various set of parameters. The numerical values of the parameters are given in Table \ref{table:2} and time is set to 50 years. We use MATLAB for the numerical simulations of the system \eqref{eq:main}. 
\begin{table}
\caption{Model Parameters}
\centering
\scalebox{0.95}{
\begin{tabular}{l c l c}
\hline \hline
Parameter & Symbol & Estimate & Source \\ [0.5ex]
\hline 
Recruitment Rate & $\Lambda$ & 280 & assumed\\
Natural death rate & $d$ & 0.01401 & estimated\\
TB induced death rate & $d_T$ & 0.1 & \cite{FengChavez,treat} \\
HIV induced death rate & $d_H$ & 0.2 & \cite{FengChavez}\\
IRIS induced death rate & $d_H^T$ & 0.33 & estimated\\
Progression rate from latent to active TB with no HIV & $k_1$ & $0.5$ & \cite{FengChavez, treat}\\
Progression rate from latent to active TB with HIV & $k_2$ & $1.3k_1$ & \cite{Silva}\\
Transition rate of TB treatment from early to late phase & $\rho_1$ & $5.56\times {10}^{-3}$ & \cite{mallela}\\
Transition rate of TB treatment from late to completion phase & $\rho_2$ & $1.11\times{10}^{-2}$ & \cite{mallela}\\
Recovery rate from TB with no HIV & $r_1$ & 0.82 & estimated\\
Per capita HIV treatment rate with no TB & $r_2$ & 0.33 & \cite{Bhunnu}\\
Per capita TB treatment rate in co-infected individuals & $r_3$ & 0.1 & \cite{mallela}\\
Per-capita TB treatment rate with no HIV & $r$ & 0.55 & \cite{Bhunnu}\\
HIV early treatment rate & $\eta_1$ & $0-0.05$ & \cite{mallela}\\
HIV late treatment rate & $\eta_2$ & $0-0.05$ & \cite{mallela}\\
Rate of occurrence of IRIS & $\gamma$ & $1 \times {10}^{-3}$ & \cite{mallela}\\ 
\hline
\end{tabular}}
\label{table:2}
\end{table}

For the numerical analysis, we use $N(0)= 20,000,\; S(0)= 12000,\; T_L(0)= 5000,\; T(I)= 1032,\;  H(0)= 340,\; H_L(0)= 113,\; C(0)= 114,\; C_1(0)= 64,\; C_2(0)= 64,\; C_1^T(0)= 32,\; C_2^T(0)= 32,\; H^T(0)= 265,\; T_T(0)= 944$ as the initial conditions. For initial conditions, it is assumed that more than half of the total population belong to the susceptible. One quarter of the total population is infected with latent TB \cite{whotb}. The population infected with HIV only is assumed to be $1.7\%$ and $78\%$ of them get proper treatment \cite{TBhiv} and $11\%$ of TB active people get co-infected with HIV. The remaining values are estimated assuming we are in controlled situation.

The natural death rate $d$ corresponds to the life expectancy of 71.4 years \cite{le} and $k_2>k_1$ implies that progression of TB is faster in co-infected individuals. It can be seen that $\beta$ and $e$ always appear together and $\beta e$ determines the TB reproduction number $\mathcal{R}_0^T$. Similarly, $\lambda$ and $\sigma$ always appear together and the product $\lambda \sigma$ determines the HIV reproduction number $\mathcal{R}_0^H$. In our calculations, we have fixed $\eta_1=0.03$ and $\eta_2=0.02$ from the given range in Table \ref{table:2} and for the Figure 4, we variate the values within range. We choose different values of $\beta e$ and $\lambda\sigma$ for our numerical simulations which are $\beta e=0.5$ for $\mathcal{R}_0^T<1$, $\beta e=2$ for $\mathcal{R}_0^T>1$, $\lambda \sigma = 0.3 $ for $\mathcal{R}_0^H<1$ and $\lambda \sigma=1$ for $\mathcal{R}_0^H>1$ resulting $\mathcal{R}_0^T=0.73$, $\mathcal{R}_0^T=2.93$, $\mathcal{R}_0^H=0.5515$ and $\mathcal{R}_0^H=1.84$, respectively.

\begin{figure}[ht]
\begin{center}
\subfloat[$\mathcal{R}_0<1$]{\includegraphics[scale=0.5]{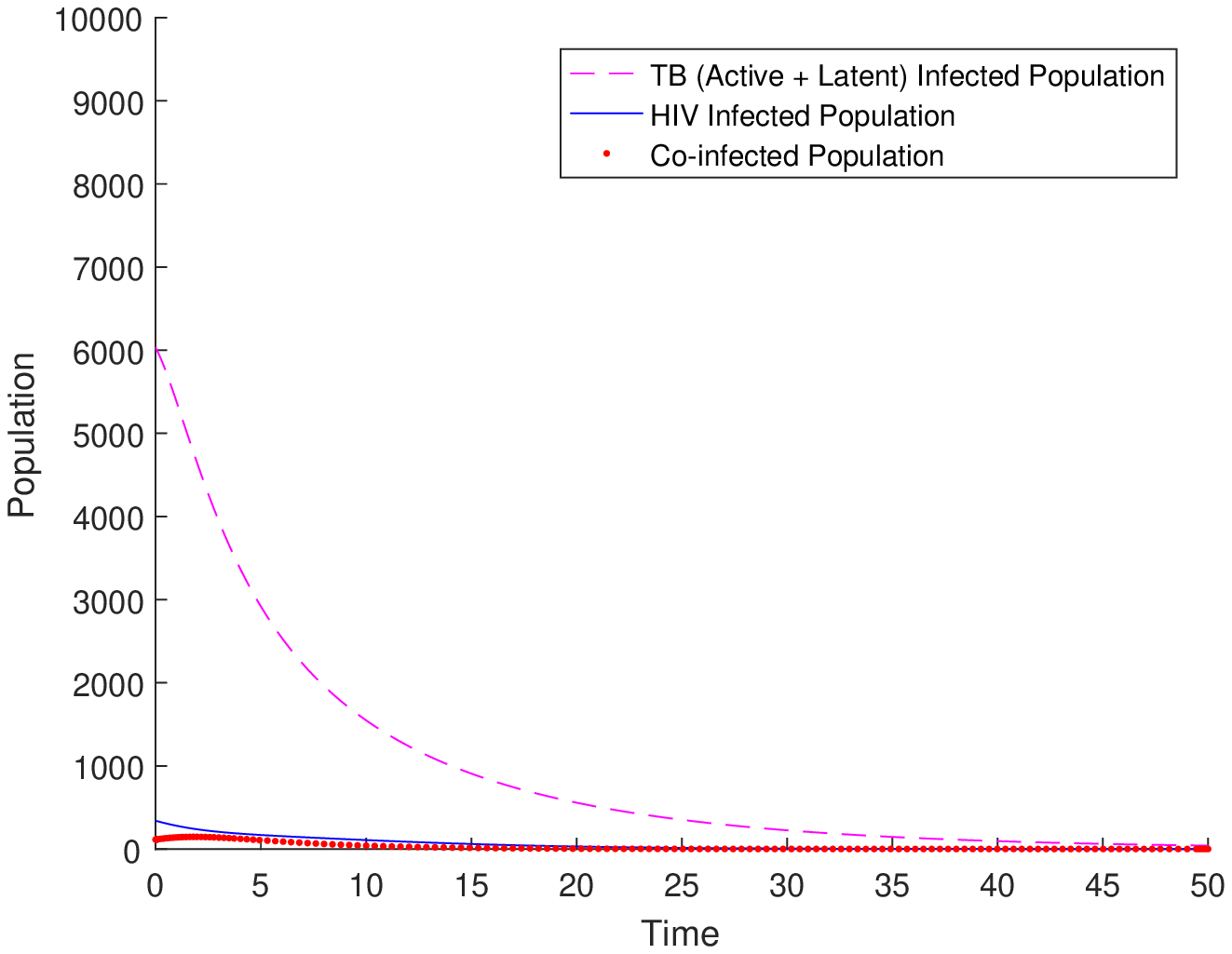}}
\subfloat[$\mathcal{R}_0^T<1$ and $\mathcal{R}_0^H>1$]{\includegraphics[scale=0.5]{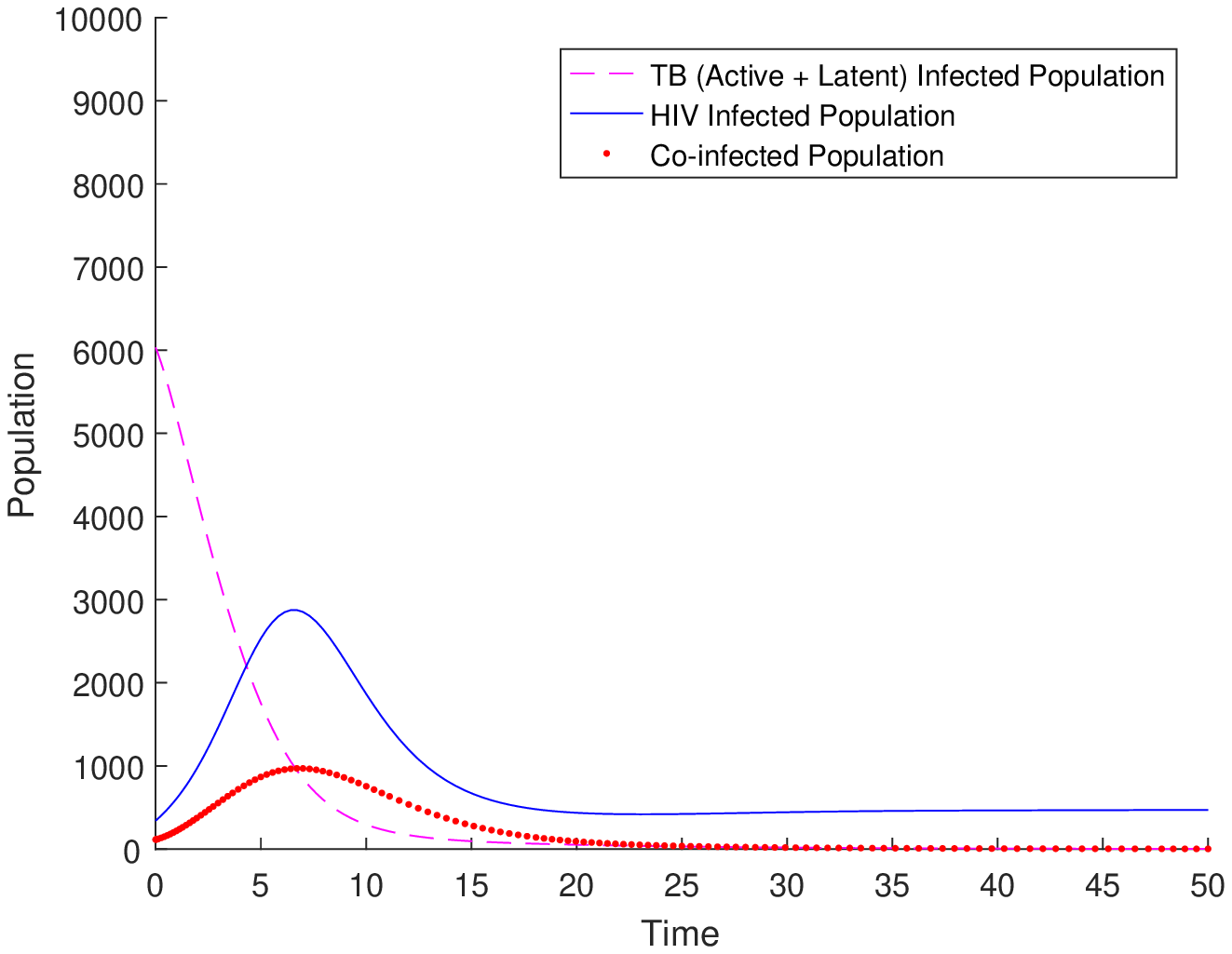}}
\hspace{0mm}
\subfloat[$\mathcal{R}_0^T>1$ and $\mathcal{R}_0^H<1$]{\includegraphics[scale=0.5]{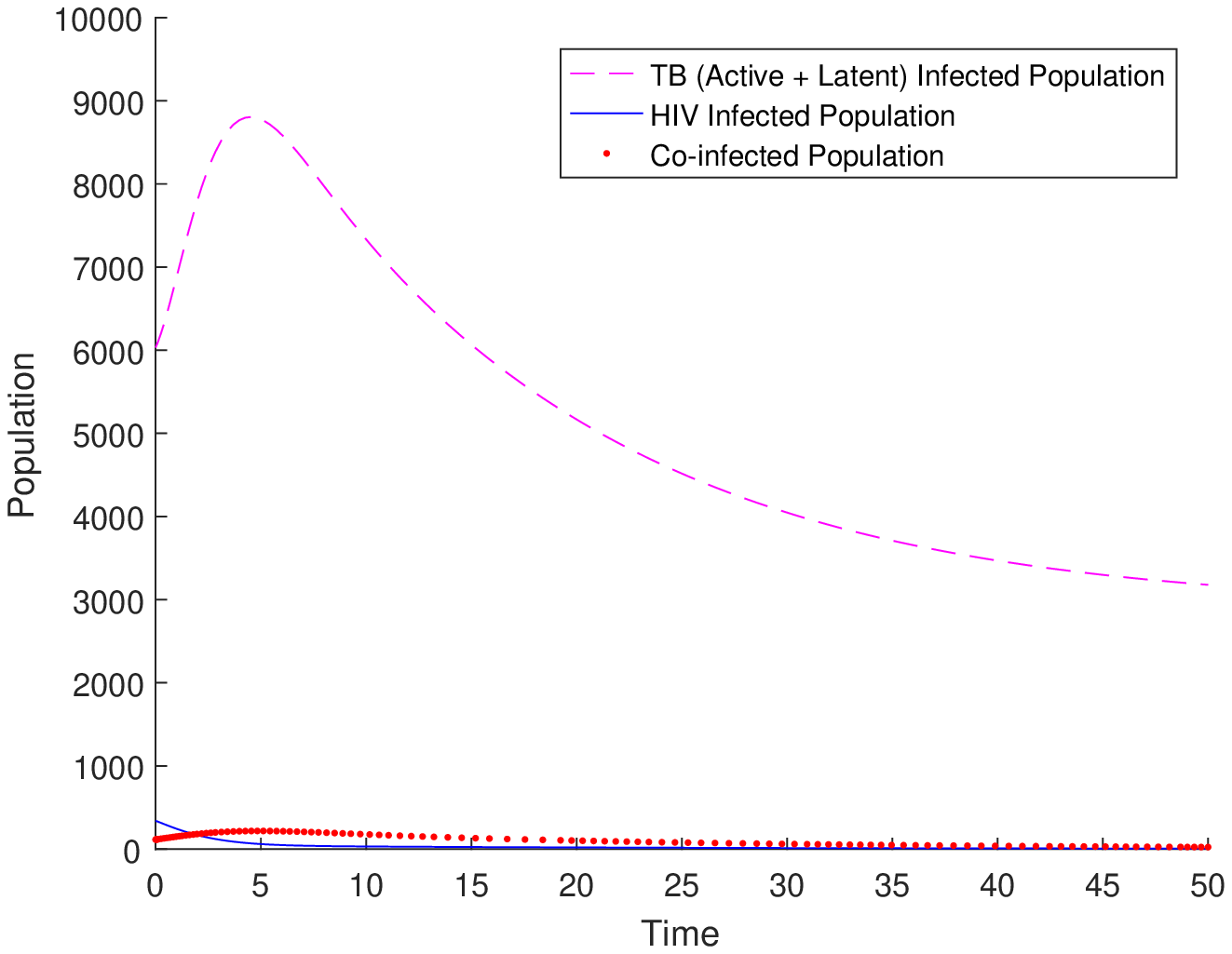}}
\subfloat[$\mathcal{R}_0^T>1$ and $\mathcal{R}_0^H>1$]{\includegraphics[scale=0.5]{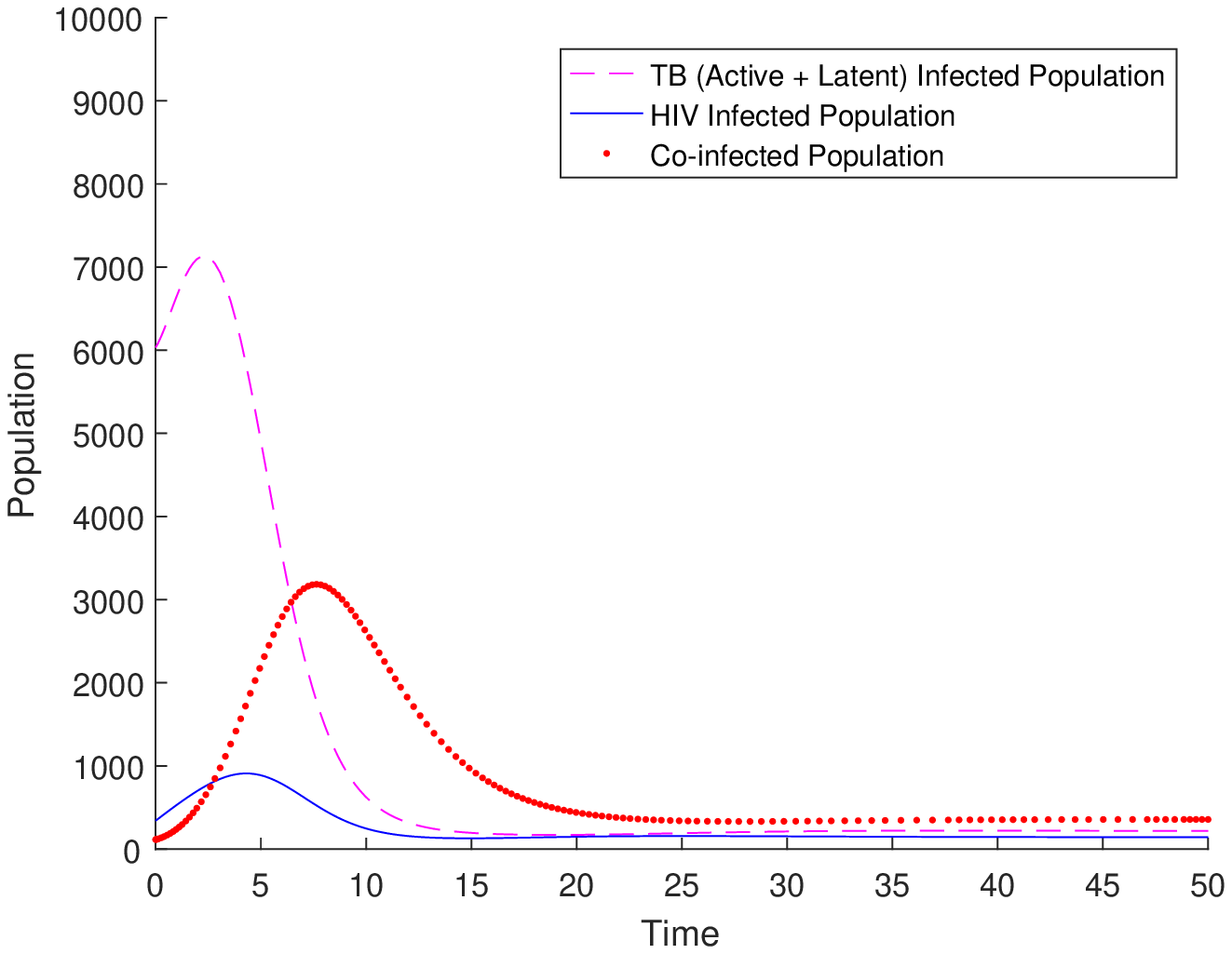}}
\hspace{0mm}
\end{center}
\caption{Effect of Reproduction Number on TB only infected population ($T_L+T_I$), HIV only infected population ($H$) and Co-infected population undergoing no treatment ($C$).}\label{fig2}
\end{figure}

 \begin{figure}
\begin{center}
\subfloat[$\mathcal{R}_0<1$]{\includegraphics[scale=0.5]{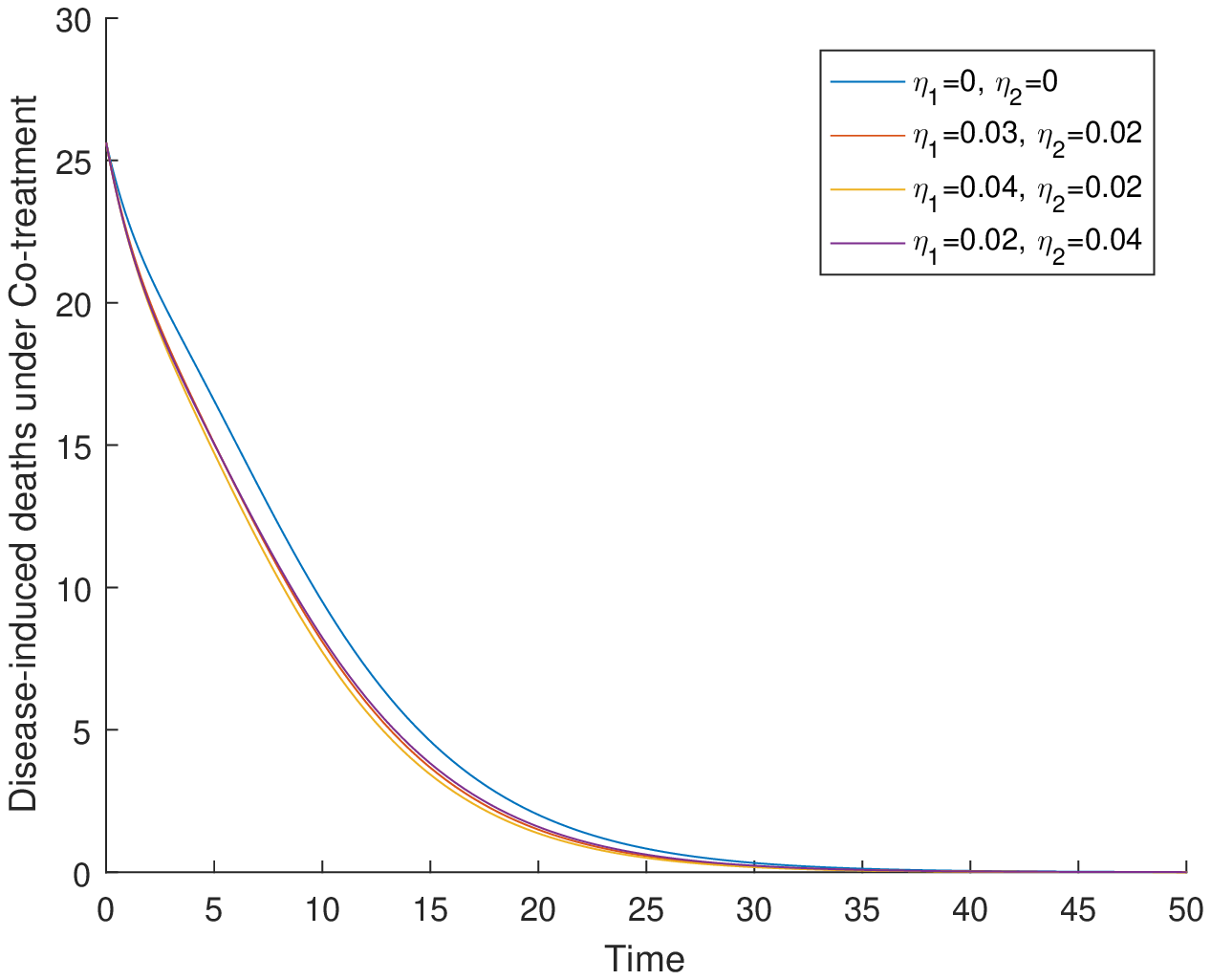}}
\subfloat[$\mathcal{R}_0^T<1$ and $\mathcal{R}_0^H>1$]{\includegraphics[scale=0.5]{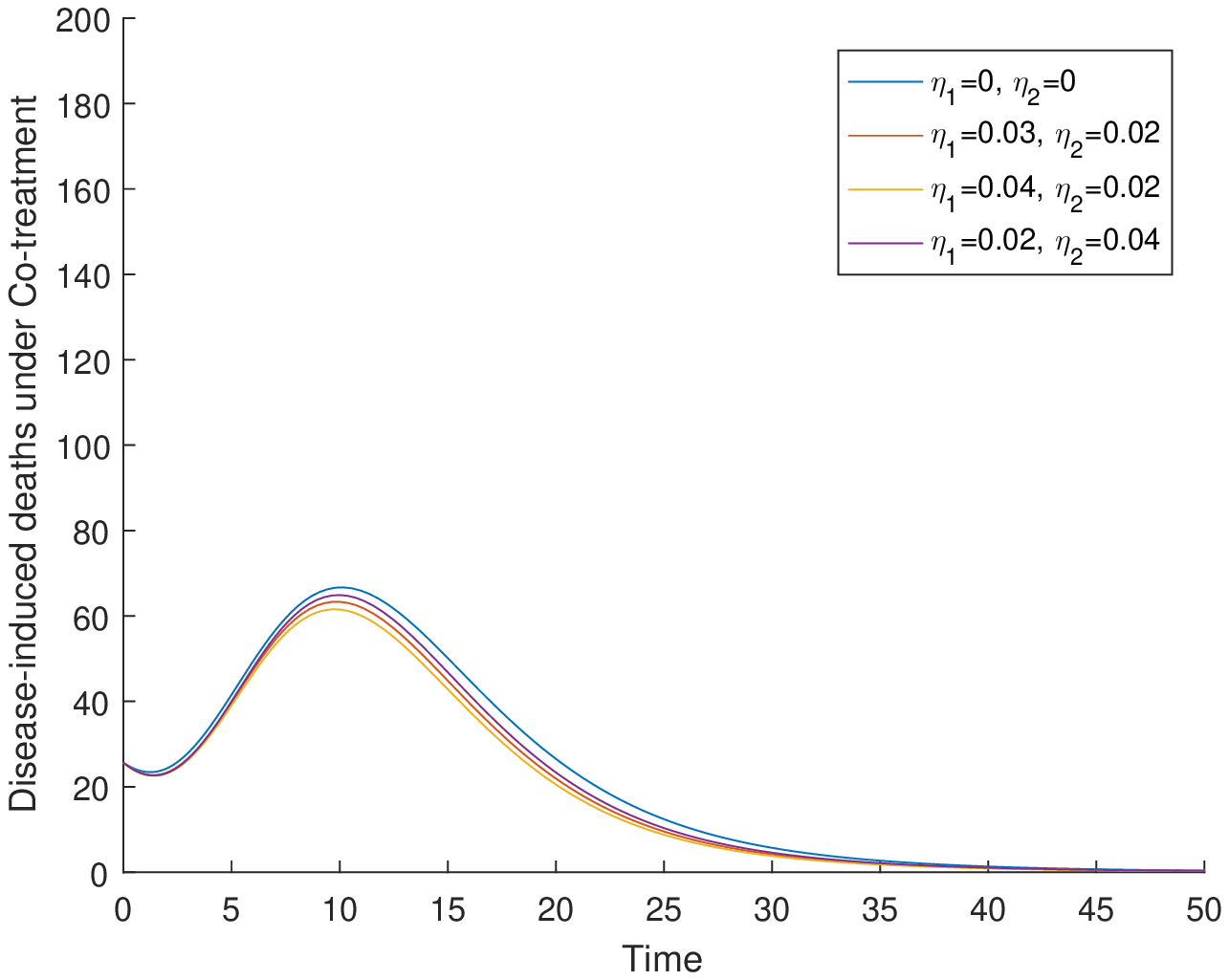}}
\hspace{0mm}
\subfloat[$\mathcal{R}_0^T>1$ and $\mathcal{R}_0^H<1$]{\includegraphics[scale=0.5]{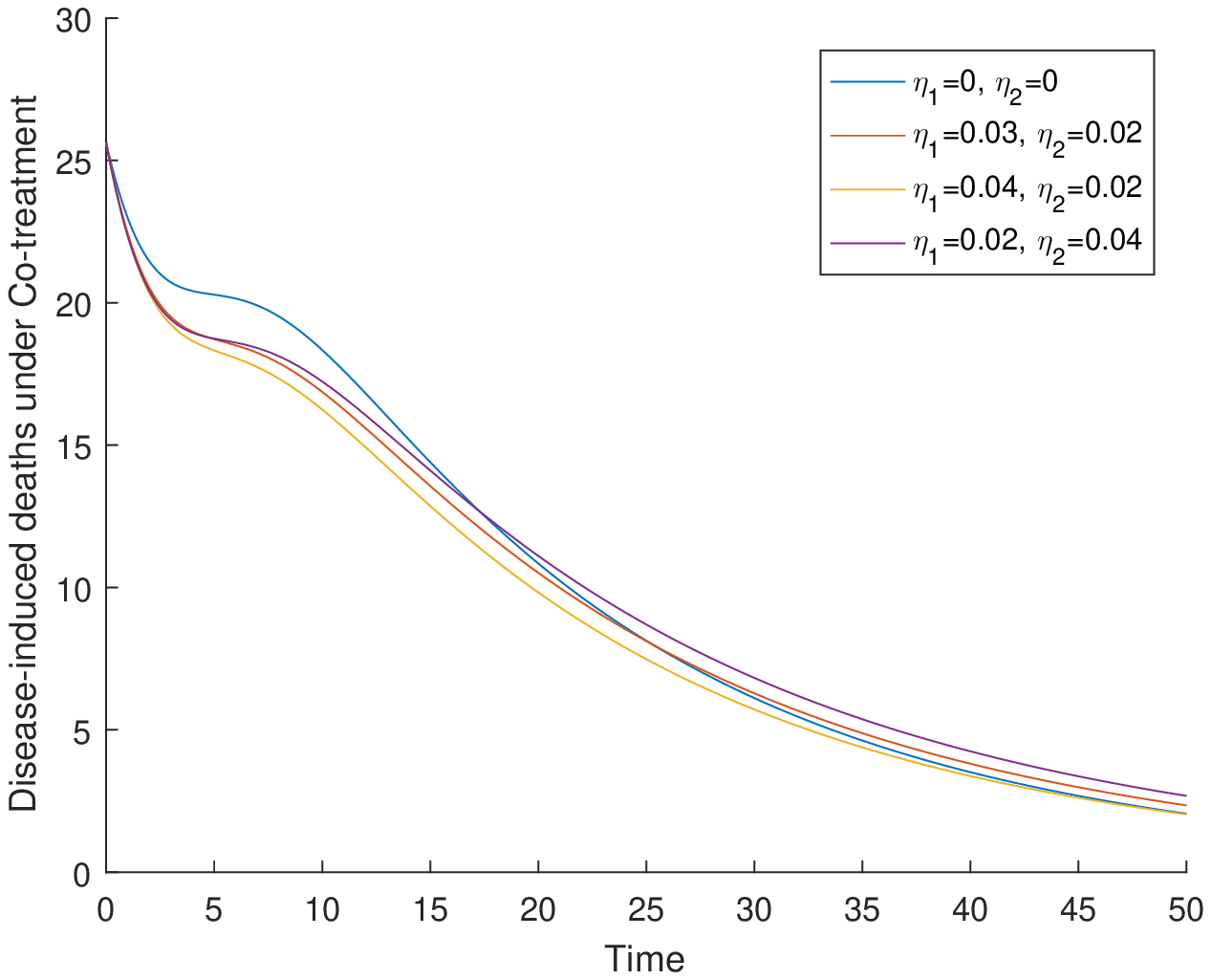}}
\subfloat[$\mathcal{R}_0^T>1$ and $\mathcal{R}_0^H>1$]{\includegraphics[scale=0.5]{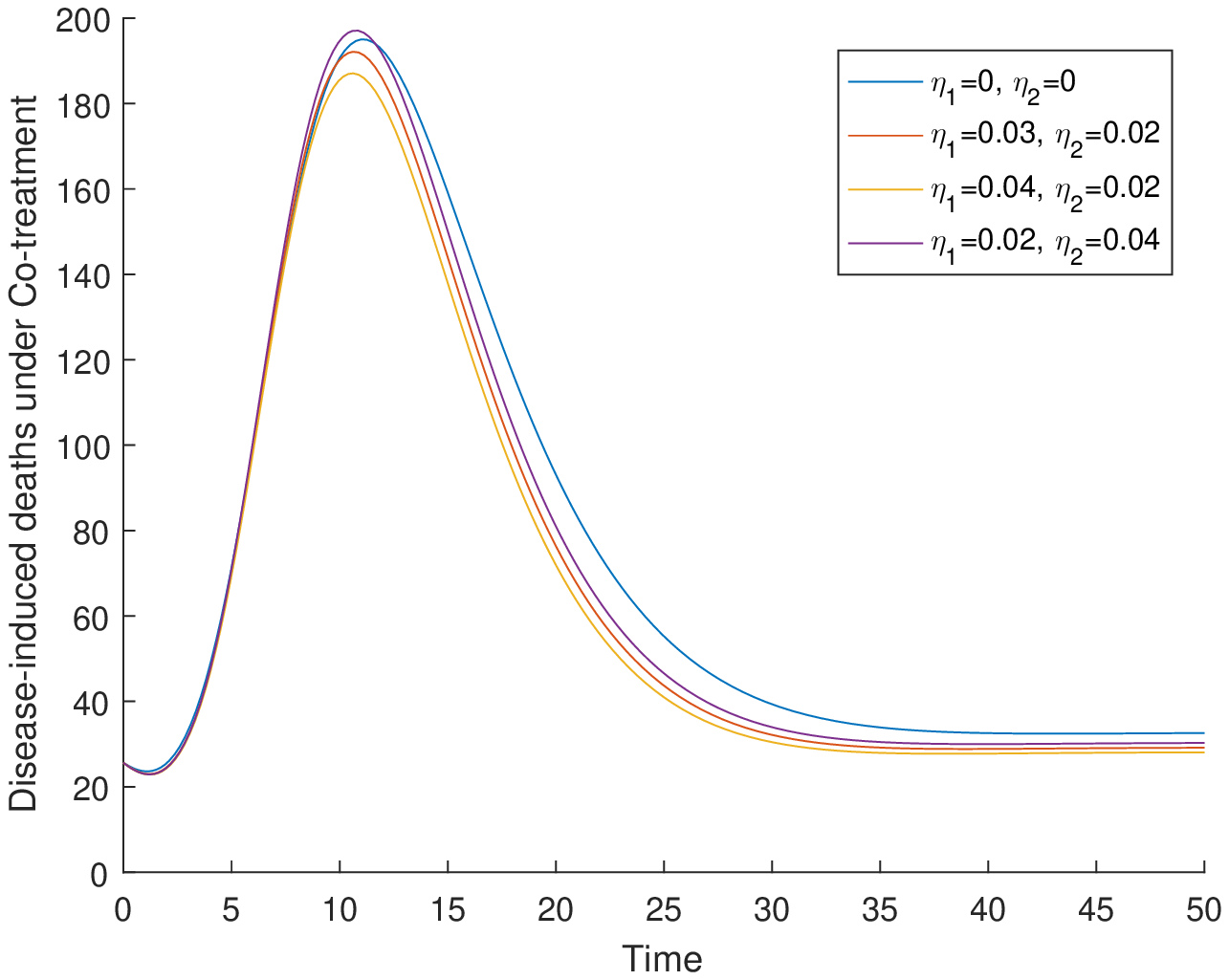}}
\hspace{0mm}
\end{center}
\caption{Effect of early or late Initiation of ART during TB treatment on disease-induced deaths of population under Co-treatment}\label{fig3}
\end{figure}
    
Figure \ref{fig2} depicts graphical representations of the change in population infected with single disease only and co-infected under no treatment regime, with change in reproduction number by plotting $T_L+T_I$, $H$ and $C$ versus time. For $\mathcal{R}_0^T<1$ and $\mathcal{R}_0^H<1$, Figure \ref{fig2} (A) shows that the diseases die out with time and approach the disease free equilibrium point $E_0$. This implies that for $\mathcal{R}_0<1$, diseases can not persist for longer duration of time. Figure \ref{fig2} (B) shows that for $\mathcal{R}_0^T<1$ and $\mathcal{R}_0^H>1$, TB infection decreases rapidly with time and finally vanishes, while HIV infected population first increases rapidly and then decreases before attaining a constant value which is $\tilde{H}$ of equilibrium point $E_H$. Co-infected population also increases very rapidly even when $\mathcal{R}_0^T<1$ and then decreases to become constant. Figure \ref{fig2} (C) shows that for $\mathcal{R}_0^T>1$ and $\mathcal{R}_0^H<1$, TB infected population first increases very rapidly and then slowly decreases to a constant value which is $\hat{T}_L+\hat{T}_I$, while HIV infected population vanishes very soon. Thus, the co-infected population also decreases with time and then vanishes. This corresponds to the equilibrium point $E_T$. Figure \ref{fig2} (D) shows that for $\mathcal{R}_0^T>1$ and $\mathcal{R}_0^H>1$, both the infections in population first increase to a maximum value and after that  decrease rapidly to attain constant values and these constant values correspond to the interior equilibrium point $E_T^H$. This represents that in favourable conditions, that is, $\mathcal{R}_0^T>1$ and $\mathcal{R}_0^H>1$, both the diseases favour each other and continue increasing rapidly and after reaching a maximum value, they again decrease to attain a constant value. This shows that no epidemic can last forever. Moreover, co-infected population is also maximum in this case.

\begin{figure}
\begin{center}
\subfloat[$\mathcal{R}_0<1$]{\includegraphics[scale=0.5]{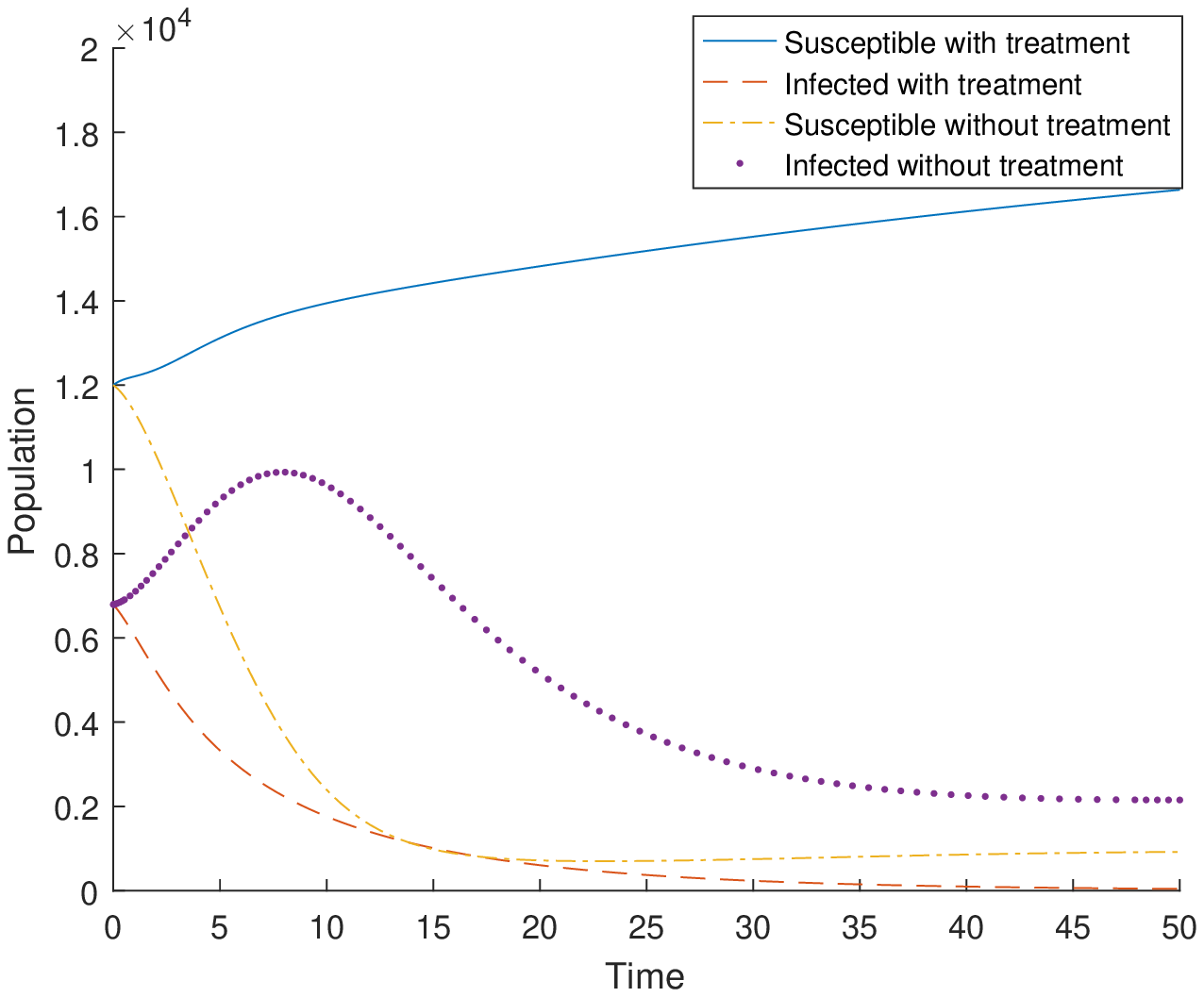}}
\subfloat[$\mathcal{R}_0^T<1$ and $\mathcal{R}_0^H>1$]{\includegraphics[scale=0.5]{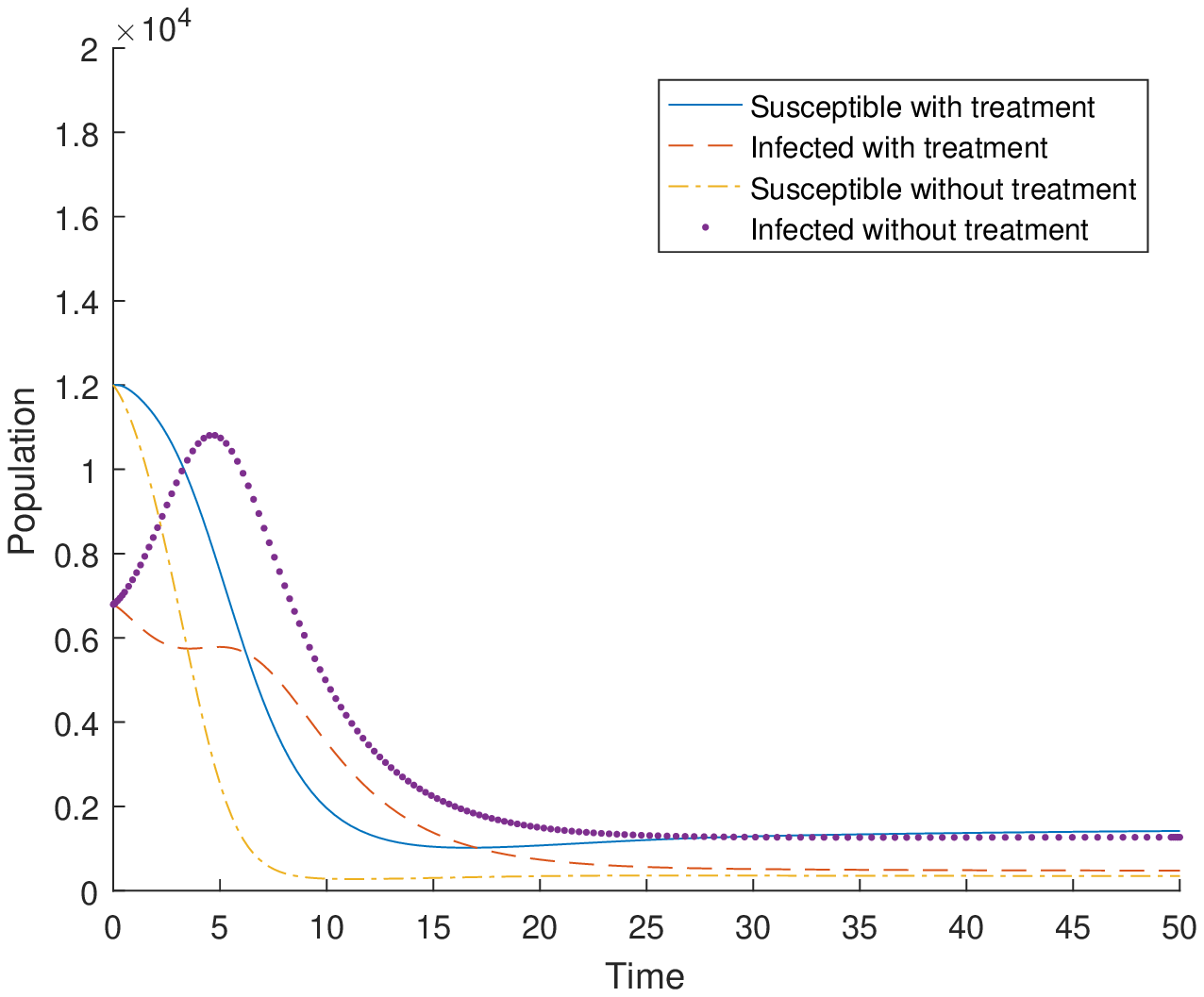}}
\hspace{0mm}
\subfloat[$\mathcal{R}_0^T>1$ and $\mathcal{R}_0^H<1$]{\includegraphics[scale=0.5]{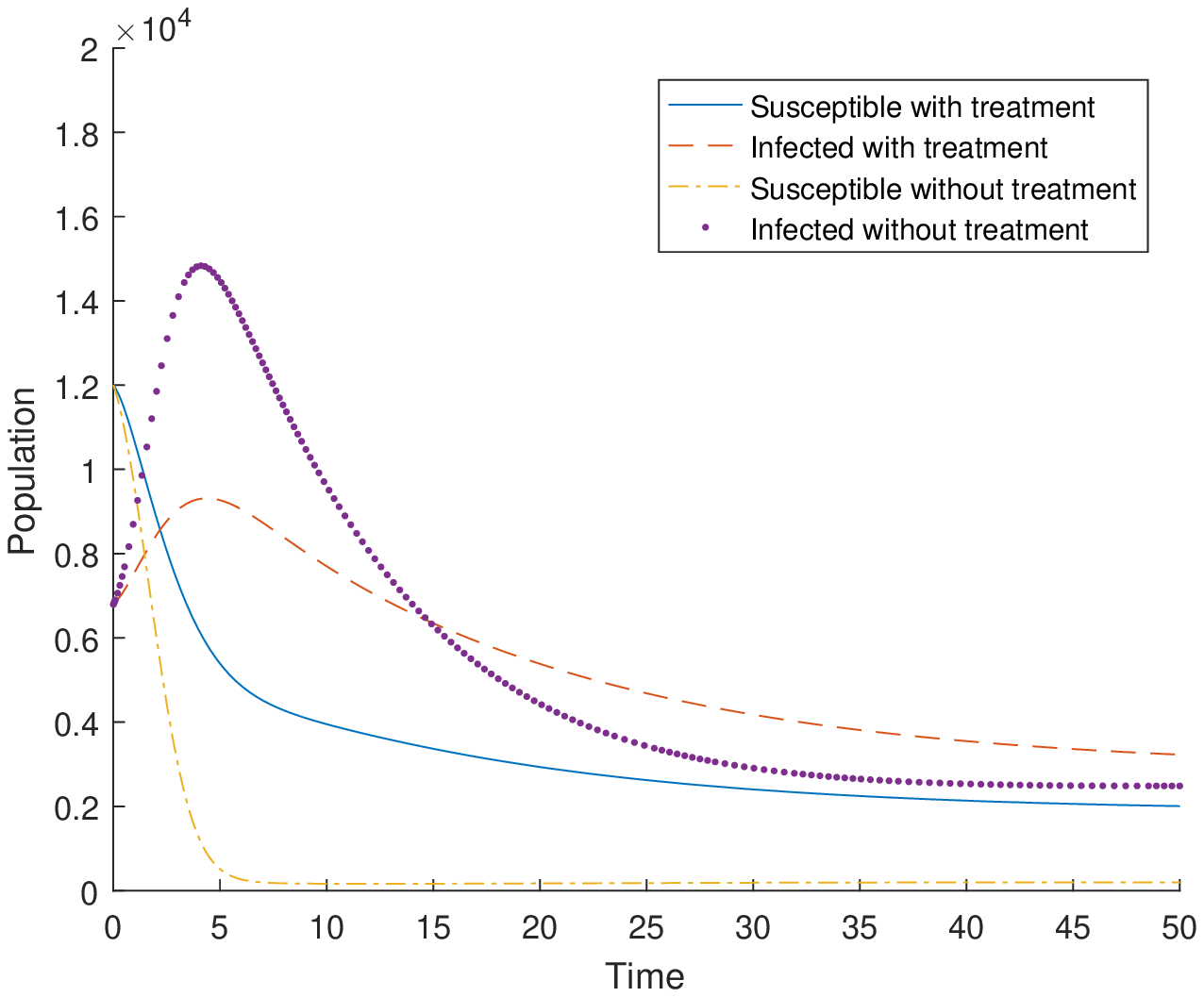}}
\subfloat[$\mathcal{R}_0^T>1$ and $\mathcal{R}_0^H>1$]{\includegraphics[scale=0.5]{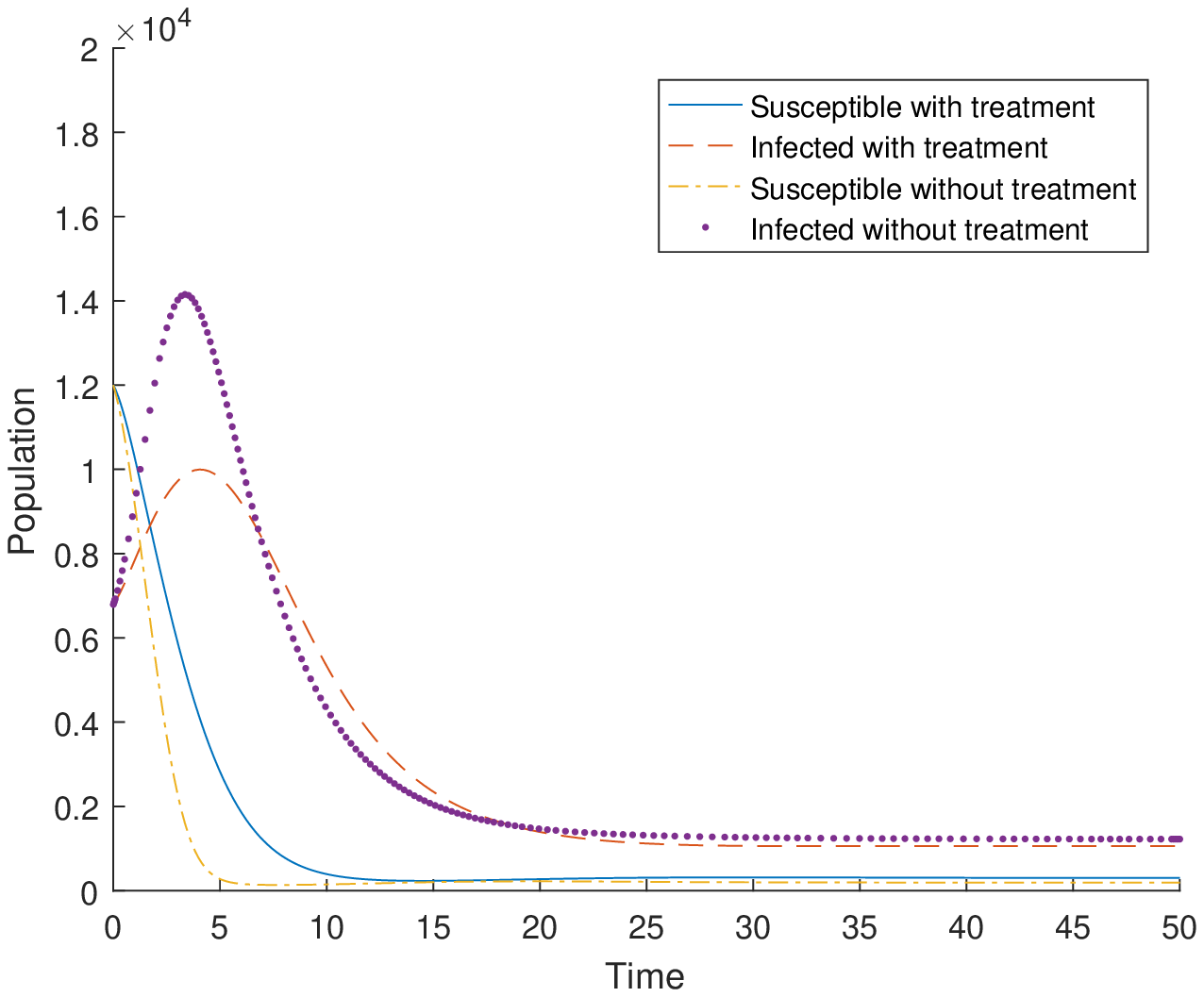}}
\hspace{0mm}
\end{center}
\caption{Effect of single disease infection treatment on the infected population ($T_L+T_I+H+H_L+C+C_1+C_2+C_1^T+C_2^T$) and susceptible ($S$).}\label{fig4}
\end{figure}

\begin{figure}[ht]
\begin{center}
{\includegraphics[scale=0.5]{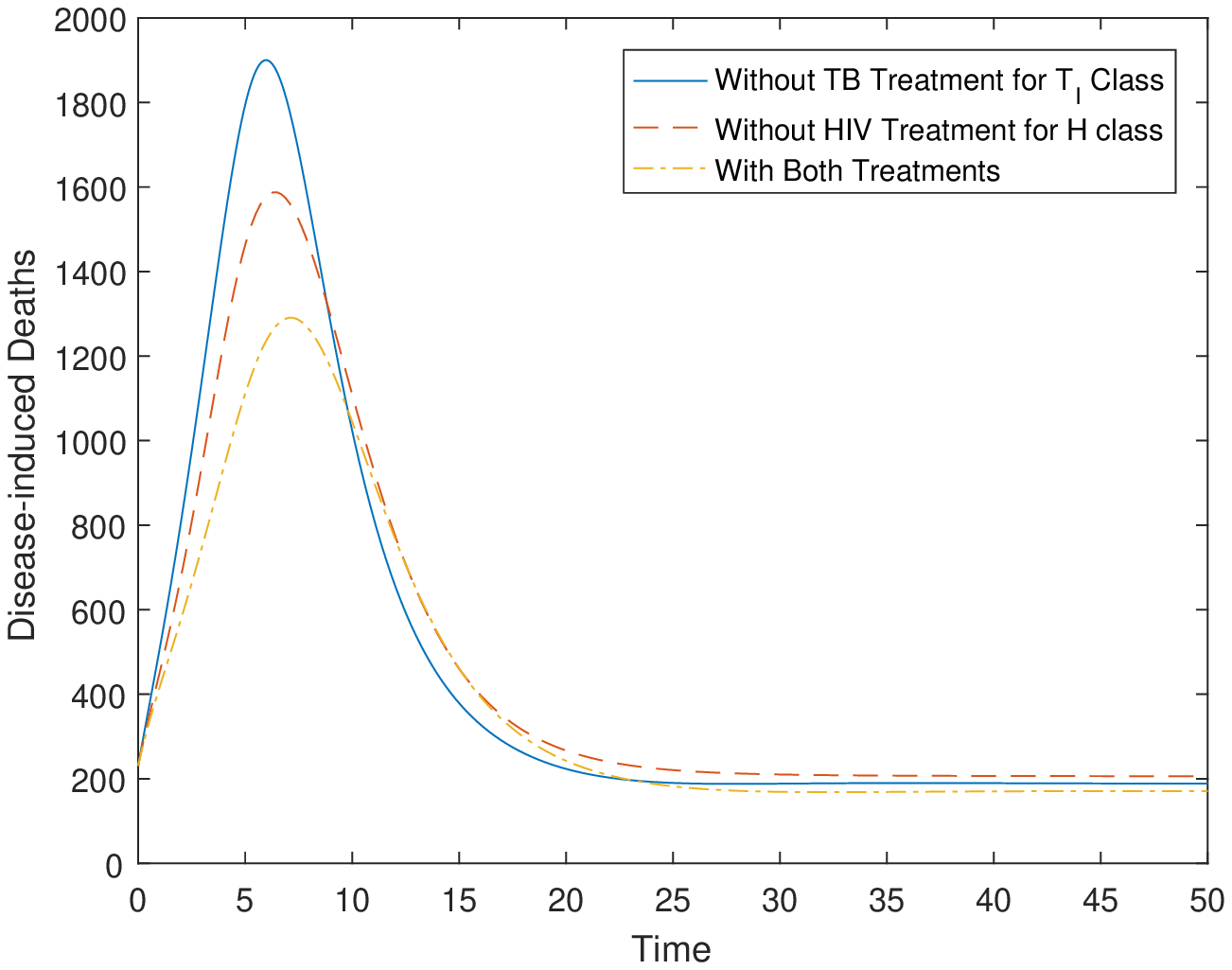}}
\hspace{0mm}
\end{center}
\caption{Effect of TB only and HIV only treatments individually on the disease-induced deaths when $\mathcal{R}_0^T>1$ and $\mathcal{R}_0^H>1$.}\label{fig5}
\end{figure}

Figure \ref{fig3} shows the effect of reproduction number with early or late initiation of ART during TB treatment on disease-induced deaths. The disease-induced deaths are the deaths among the population caused by the diseases other than natural deaths. We have plotted the time versus disease-induced deaths in compartments $C_1^T$ and $C_2^T$ which is $(d_H+\gamma d_H^T)C_1^T+(d_H+\alpha d_H^T)C_2^T$. We have plotted the graphs for four different set of values $\eta_1=0$, $\eta_2=0$; $\eta_1=0.03$, $\eta_2=0.02$ ; $\eta_1=0.04$, $\eta_2=0.02$ and $\eta_1=0.02$, $\eta_2=0.04$. Figure \ref{fig3} shows that higher the rate of early phase HIV treatment during TB treatment, lesser are the disease-induced deaths. Figure shows that for $\mathcal{R}_0<1$ disease-induced deaths decrease with time and vanish after some time. For $\mathcal{R}_0^T<1$ and $\mathcal{R}_0^H>1$, the disease-induced deaths decrease to vanish after attaining a maximum value. For $\mathcal{R}_0^T>1$ and $\mathcal{R}_0^H<1$ the disease-induced deaths increase after a slight decrease and then keep on decreasing. For $\mathcal{R}_0>1$, the disease-induced deaths increase to a maximum value and then decrease to attain a constant value. Figure \ref{fig3} shows that higher the rate of HIV early treatment during TB treatment, lesser are the disease-induced deaths, while increased rate of HIV late treatment does not have very remarkable impact on disease-induced deaths. 

The total number of disease induced deaths in different cases are given in Table \ref{table:3}. Table \ref{table:3} shows that for different reproduction numbers the total number of disease-induced deaths are lesser for higher value of $\eta_1$ while $\eta_2$ does not affect it in similar manner.  

\begin{table}[H]
\caption{Total disease-induced deaths in span of 50 years for different values of $\eta_1$ and $\eta_2$}
\centering
\scalebox{0.85}{
\begin{tabular}{l c  c c c}
\hline \hline
Reproduction number & \makecell{Deaths when \\$\eta_1=0,\; \eta_2=0$} & \makecell{Deaths when \\ $\eta_1=0.03,\; \eta_2=0.02$ } & \makecell{Deaths when \\ $\eta_1=0.04,\; \eta_2=0.02$ } & \makecell{Deaths when \\ $\eta_1=0.02,\; \eta_2=0.04$} \\ [0.5ex]
\hline 
$\mathcal{R}_0<1$ & 228 & 203 & 196 &  204\\
$\mathcal{R}_0^T<1$ and $\mathcal{R}_0^H>1$ & 1096 & 990 & 954 &  1023\\
$\mathcal{R}_0^T>1$ and $\mathcal{R}_0^H<1$ & 507 & 490 & 465 &  511\\
$\mathcal{R}_0^T>1$ and $\mathcal{R}_0^H>1$ & 3678 & 3350 & 3228  &  3473\\
\hline
\end{tabular}}
\label{table:3}
\end{table}

Figure \ref{fig4} shows the effect of treatment for single disease infection, that is, effect of $r$, $r_2$ and $r_1$ on susceptible and the total infected population ($T_L+T_I+H+H_L+C+C_1+C_2+C_1^T+C_2^T$). We plot the graph of the time versus susceptible and infected population for different reproduction numbers. For population without treatment, we assume $r = r_1 = r_2 = 0$, that is, there is no treatment for TB for $T_I$ class, so no recovery by treatment i.e. $r_1=0$ and there is no treatment for $H$ compartment. Figure \ref{fig4} (A) shows that for basic reproduction number $\mathcal{R}_0<1$, the infection dies out with time when treatment is considered while for no single disease treatment infected population does not vanish even in a span of 50 years. The susceptible tend to increase with treatments while for the other case susceptible decrease to attain a constant value which is very small. Figure \ref{fig4} (B) shows that when $\mathcal{R}_0^T<1$ and $\mathcal{R}_0^H>1$, for no treatment of single disease infected population, the infected population first increases rapidly and then decreases to become constant at a higher value than with the treatment. The susceptible population undergoing treatment for single disease infection as well as co-infection treatment is always greater than the susceptible population with only co-infection treatment and with passage of time treatment increases the susceptible population while with no treatment, susceptible population decreases to a very small value. In our discussion henceforth, the term treatment signifies the treatments for single disease infected population. Figure \ref{fig4} (C) shows that for $\mathcal{R}_0^T>1$ and $\mathcal{R}_0^H<1$, infected population first increases in both cases and then decreases, but for no treatment infection increases rapidly than for the treatment and then decreases to attain a constant value less than the value attained by the curve with treatment. The susceptible in both cases decrease to attain a constant value but susceptible population without treatment decrease very rapidly to attain an extremely low value as compared to susceptible with treatment. Figure \ref{fig4} (D) shows that for $\mathcal{R}_0^T>1$ and $\mathcal{R}_0^H>1$, infected population without treatment increases more rapidly and then decreases to attain a constant value a little larger than for the population with treatment. The susceptible population decreases rapidly with time so that population without treatment is always less than the population with treatment. 

We conclude that in absence of treatment for single disease infection, the disease-induced deaths increase and infection persists even when $\mathcal{R}_0<1$, disease-induced deaths increase independent of the reproduction number and the susceptible population decreases to a very small quantity.

Figure \ref{fig5} shows the effect of single disease treatments individually on the disease-induced deaths when diseases are epidemic where disease-induced deaths are the deaths which are caused due to diseases. It represents the disease-induced deaths in all the three cases. It can be clearly seen that the  disease-induced deaths are maximum when there is no treatment for TB only infected individuals while it is lesser in the case of no treatment for HIV infected individuals followed by the case when both the treatments are given. This figure shows that TB only treatment is most important to reduce the infection and disease-induced deaths. While the statistics are best when both the treatments are given to the population. Moreover, this figure also shows that if diseases are epidemic then treatment for one disease only, that is, treatment for  TB only or HIV only infected individuals is not sufficient for the eradication of diseases from the population.

\section{Conclusion}\label{Section7}
The main model \eqref{eq:main} is a 12 dimensional system which focuses on need of single disease treatment in addition to co-infection treatment. The TB only and HIV only models have globally stable disease free equilibria when their corresponding reproduction number is less than unity. For reproduction number greater than unity, the endemic equilibria also exist and are locally asymptotically stable. The full HIV-TB co-infection model is shown to have a locally asymptotically stable disease free equilibrium when $\mathcal{R}_0<1$. The HIV only and TB only equilibrium exist and are locally asymptotically stable when $\mathcal{R}_0^H>1$, $\mathcal{R}_0^T<1$ and $\mathcal{R}_0^H<1$, $\mathcal{R}_0^T>1$, respectively. Numerical simulations indicate the presence of interior equilibrium for $\mathcal{R}_0^H>1$, $\mathcal{R}_0^T>1$. The system undergoes supercritical transcritical bifurcation when $R_0^T=1$ and $R_0^H=1$ whereas $\beta^*=\beta e$ and $\lambda^*=\lambda \sigma$ act as the bifurcation parameters respectively.

The simulation results provided many interesting insights into the effect of the dynamics of HIV-TB co-infection. Figure \ref{fig2} shows that the presence of TB may have a significant influence on HIV dynamics. For endemic TB, prevalence of HIV increases. When HIV is endemic that is $\mathcal{R}_0^H>1$ then even for $\mathcal{R}_0^T<1$, the co-infected population increases dramatically. Figure \ref{fig3} shows that early initiation of ART during TB treatment is more effective to reduce disease-induced deaths while late initiation doesn't have very remarkable impact on it. Figure\ref{fig4} shows that co-infection treatment alone is not sufficient to eradicate the diseases, treatment for TB only and HIV only patients separately is also necessary. In the absence of that disease-induced deaths become very high and infection prevails even when reproduction number is less than unity. Numerical results show that investing more in single disease infection treatments is more effective to reduce the infection and disease-induced deaths. Figure \ref{fig5} emphasises the role of TB only treatment in reducing the infection in population and decreasing the disease-induced deaths when both the diseases are epidemic. Investing more in the TB treatment programs can be a better approach to control the disease dynamics as it can be completely cured and its duration is short. Thus, declining the HIV prevalence. Moreover, TB is more contagious than HIV and a single infection can cause many secondary infections. Hence, controlling TB infection can be an important aspect in controlling the co-infection dynamics but treatment for one disease only is not sufficient for the complete eradication of the diseases from the population. 
 
\section*{Acknowledgments} The authors would like to thank the anonymous referees for their extensive comments on the revision of the manuscript which really improved the quality of the paper. The author, Shikha Jain is financially supported by the University Grant Commission (UGC), Government of India (Sr. No. 2061440971). She gratefully acknowledges the support for the research work.

\end{document}